\documentclass[journal]{IEEEtran}
\IEEEoverridecommandlockouts
\usepackage{verbatim}
\usepackage{epsfig}
\usepackage{amsmath}
\usepackage{lipsum}
\usepackage{amsthm}
\usepackage{amssymb}
\usepackage{psfrag}
\usepackage{soul}
\usepackage[T1]{fontenc}
\usepackage[ruled]{algorithm2e}
\usepackage{mathtools, nccmath}
\usepackage{cuted}
\usepackage{multirow}
\usepackage[table,xcdraw]{xcolor}
\usepackage{etoolbox}
\usepackage{textcomp}
\usepackage{xcolor}
\usepackage{url}
\usepackage{dsfont}
\usepackage{mathtools}
\usepackage{bbm}

\usepackage[usenames,dviIBRames]{pstricks}
\usepackage{etoolbox}
\usepackage{graphicx}
\graphicspath{{Images/}}
\ifCLASSOPTIONcompsoc
\usepackage[caption=false,font=footnotesize,labelfont=sf,textfont=sf]{subfig}
\else
\usepackage[caption=false,font=footnotesize]{subfig}
\fi
\usepackage{stfloats}
\makeatletter
\patchcmd{\@makecaption}
  {\scshape}
  {}
  {}
  {}
\makeatother
\usepackage[noadjust]{cite}

\usepackage[font=footnotesize,skip=0pt]{caption}

\newtheorem{theorem}{Theorem}[]
\newtheorem{lemma}{Lemma}[]

\newcommand{\admm}{\textbf{\texttt{DC-DistADMM}}}
\newcommand{\ii}{\mathrm{i}}
\newcommand{\jj}{\mathrm{j}}
\newcommand{\s}{\mathrm{s}}
\newcommand{\kk}{\mathrm{k}}

\newcommand{\sst}{\mathrm{ss}}
\newcommand{\Vnom}{V_\mathrm{nom}}
\newcommand{\Qnom}{Q^\mathrm{nom}}
\newcommand{\unom}{u^\mathrm{nom}}

\newcolumntype{L}[1]{>{\raggedright\arraybackslash}p{#1}}

\DeclareMathOperator*{\argmin}{argmin}
\DeclareMathOperator*{\minimize}{minimize}

\definecolor{color1}{rgb}{0,0,0}

\begin{document}
\title{Analysis of a Distributed Optimization-Based Control Architecture for Inverter-Interfaced \\ Virtual Power Plants}
\author{Vivek Khatana$^{\dagger}$,
        Soham Chakraborty$^{\ddagger}$,
        Murti V. Salapaka$^{\wr}$
\thanks{$^{\dagger}$Department of Mechanical Science and Engineering, University of Illinois at Urbana-Champaign, IL, USA {\tt\scriptsize \{vkhatana\}@\{illinois.edu\}}, $^{\ddagger}$Department of Electrical Engineering, Indian Institute of Science, Karnataka, India {\tt\scriptsize \{schakraborty@iisc.ac.in\}}, $^{\wr}$Department of Electrical and Computer Engineering, University of Minnesota, MN, USA {\tt\scriptsize \{murtis@umn.edu\}}}
\thanks{The research was conducted when Vivek Khatana and Soham Chakraborty were at the University of Minnesota, Twin Cities, with the support of the United States Department of Energy via grant~DE-CR$0000040$.
}
}
\maketitle



\section{The Virtual Power Plant Control Setup}\label{sec:VPP_description}
We consider secondary-level control~\cite{Hierarchical_Control} for reactive power sharing and voltage and frequency regulation in the VPP. The $P$-$\omega$/$Q$-$V$ droop control law~\cite{mu_syn} for the grid forming inverter-interfaced distributed generation (GFM-IIDG) is given by:
\begin{align}
    \omega_\mathrm{i} &= \omega_\mathrm{nom} - r_{\omega_\mathrm{i}} P_\mathrm{i}^\mathrm{avg}, \ \ P_\mathrm{i}^\mathrm{avg} =  \frac{1}{(\tau_{P_\ii} s+1)}P_\mathrm{i}, \label{eq:w_droop}\\
     V_\ii &= \Vnom - r_{V_\ii}Q_\ii^\mathrm{avg} + U_\ii^\star,\ \ Q_\ii^\mathrm{avg} =  \frac{1}{(\tau_{Q_\ii} s+1)}Q_\ii \label{eq:v_droop_0}
\end{align}
where $\omega_\mathrm{nom}$, $\omega_\mathrm{i}, \Vnom, V_\ii$ are the nominal and measured frequency and voltages, respectively, and $r_{\omega_\mathrm{i}}, r_{V_\ii} \in \mathbb{R}_{>0}$ are the droop coefficients. $P_\mathrm{i}$ and $Q_\ii$ are the instantaneous active and reactive power given by 
\begin{align}
    \hspace{-0.1in} P_\mathrm{i} &= G_\mathrm{ii} V_\ii^2 - \sum_{\kk \in N_\mathrm{i}} B_\mathrm{ik} V_\ii V_\mathrm{k} (\delta_\mathrm{i} - \delta_\mathrm{k}) - P_\mathrm{i}^\mathrm{gfl} \label{eq:activepower},\\
    \hspace{-0.12in}  Q_\ii & = - B_\mathrm{ii} V_\ii^2 +  \sum_{\kk \in N_\mathrm{i}} B_\mathrm{ik} V_\ii V_\mathrm{k}.\label{eq:reactivepower}
\end{align}
Here, the complex admittance, $Y_\mathrm{ik} = Y_\mathrm{ki} \in \mathbb{C}$, between buses $\ii \in \{1,2,\dots,n_{\mathrm{gfm}}\}$ and $\kk \in \{1,2,\dots,n_{\mathrm{gfm}}\}$ are represented as $Y_\mathrm{ik} := -j B_{\ii \kk} = -j B_{\kk \ii} =: Y_\mathrm{ki}$ with $B_{\ii \kk} = B_{\kk \ii} < 0$. The set of neighbors of a bus $\ii \in \{1,2,\dots,n_{\mathrm{gfm}}\}$ is defined as $N_\mathrm{i}:= \{\kk | \kk \in \{1,2,\dots,n_{\mathrm{gfm}}\}, \kk \neq \ii, Y_\mathrm{\ii\kk} \neq 0 \}$ and $Y_{\ii \ii} = G_{\ii \ii} + j (B^{\mathrm{sh}}_{\ii} + \sum_{\kk \in N_\mathrm{i}} B_\mathrm{ik}) = G_{\ii \ii} + j B_{\ii \ii}$ where, $G_{\ii \ii} > 0$ is the shunt conductance, $B^{\mathrm{sh}}_{\ii} < 0$ is the shunt susceptance and $B_{\ii \ii} < 0$. $P_\mathrm{i}^\mathrm{avg}$ and $Q_\ii^\mathrm{avg}$ in~\eqref{eq:w_droop}-\eqref{eq:v_droop_0} are the averaged active and reactive power injections of GFM-IIDG $\ii$, obtained by passing the instantaneous active and reactive power $P_\mathrm{i}$ and $Q_\ii$ through low-pass filters with time constants, $\tau_{P_\ii}, \tau_{Q_\ii} \in \mathbb{R}_{>0}$ respectively. 

Signals $U_\ii^\star(t), \ii \in \{1,2,\dots, n_{\mathrm{gfm}}\}$ are determined using the following design law, 
\begin{align}\label{eq:U_i_design_law}
    U_\ii^\star(t) := & \ u_\ii(t) - \beta_{Q_\ii} r_{V_\ii} Q^\mathrm{avg}_\ii(t) + \beta_{V_\ii} (\Vnom - V_\ii(t)),
\end{align} 
where $\beta_{V_\ii}, \beta_{Q_\ii} \in \mathbb{R}$ are design hyper-parameters. With the compensation signals $U_\ii^\star$ in~\eqref{eq:U_i_design_law} the control law~\eqref{eq:v_droop_0} becomes
\begin{align}\label{eq:voltage_reference}
    V_\ii &=  V_\mathrm{nom} -  r_{V_\ii}\frac{(1 + \beta_{Q_\ii})}{(1 + \beta_{V_\ii})} Q_\ii^\mathrm{avg} + \frac{1}{(1 + \beta_{V_\ii})}u_\ii,
\end{align}
where, $u_\ii(t), t\in(t_\s,t_{\s+1}]$, is updated from the last value $u_\ii(t_\s)$ as follows\footnote{The sampling instants satisfy $t_{\s+1}:=t_\s+\Delta_\s+\Delta$ for $\s\in\{0,1,2,\dots\}$}
\begin{equation*}
     u_\ii(t) = \hspace{-0.035in} \Bigg\{  \hspace{-0.02in} \begin{aligned} 
     & u_\ii(t_\s), t \in (t_\s, t_\s + \Delta_\s),\\ 
     & \textstyle u_\ii(t_\s) + \left[\frac{(t - t_\s - \Delta_\s)(x_{\s,\ii} - u_{\ii}(t_\s))}{\Delta}\right], t \in (t_\s + \Delta_\s, t_{\s+1}]
    \end{aligned} 
\end{equation*}
for all $\ii\in\{1,\dots,n_{\mathrm{gfm}}\}$.  Here $x_{\s,\ii}$ denotes the solution of the following problem instantiated with measurements $[V_\ii(t_\s), \ Q^\mathrm{avg}_\ii(t_\s)] \in \mathbb{R}^2, \ii \in \{1,\dots, n_{\mathrm{gfm}}\}$:
\begin{equation}\label{eq:gfm_optimization}
\begin{aligned}
    x_\s &:= \argmin_{x_{\s,1}, \dots, x_{\s,n_\mathrm{gfm}}}  \sum_{\ii=1}^{n_\mathrm{gfm}} \ \frac{1}{2}(x_{\s,\ii} - \alpha_\ii (t_\s))^2 \\
    \mbox{subject to} \ x_{\s,\ii} &= x_{\s,\jj}, \  \forall \ii,\jj \in \{1,\dots,n_\mathrm{gfm}\}, \\
    x_{\s,\ii} &\in D_\ii^{V_\Delta}(t_\s), \  \forall \ii \in \{1,\dots,n_\mathrm{gfm}\}, \\ 
     \alpha_\ii (t_\s) &:= (1+\beta_{Q_\ii})(\Vnom - V_\ii(t_\s))  \\ 
    & \hspace{0.15in} - \beta_{V_\ii} r_{V_\ii} Q^\mathrm{avg}_\ii (t_\s) \ \forall \ii \in \{1,\dots,n_\mathrm{gfm}\}.
\end{aligned}
\end{equation}
Here, $D_\ii^{V_\Delta}(t_\s) := \big\{ \zeta \big | | \zeta - (1+\beta_{Q_\ii}) r_{V_\ii} Q^\mathrm{avg}_\ii (t_\s) + \beta_{V_\ii} (\Vnom - V_\ii(t_\s)) | \leq V_\Delta \big\}$ with a prescribed hyper-parameter $V_\Delta \in \mathbb{R}_{>0}$.

In addition,  the secondary-level controller designs active power references $P_\ii^\star$, $\ii \in {1,\dots,n_{\mathrm{gfl}}}$, for the primary-level $\mathrm{PQ}$-dispatch controllers of the grid-following (GFL)-IIDGs as
\begin{align}\label{eq:Pi_star_design_law}
P_\ii^\star(t) := P_\ii^{\min} + p_\ii(t)(P_\ii^{\max} - P_\ii^{\min}).
\end{align}
Here, $P_\ii^{\min}$ and $P_\ii^{\max}$ denote the minimum and maximum active power generation capacities of GFL-IIDG $\ii$, respectively, and $p_\ii(t),t\in(t_\s,t_{\s+1}]$, are updated from $p_\ii(t_\s)$ as
\begin{equation*} 
    p_\ii(t) = \hspace{-0.035in} \Bigg\{  \hspace{-0.02in} \begin{aligned}
            & p_\ii(t_\s),  t \in (t_\s, t_\s + \Delta_\s),\\
          & \textstyle p_\ii(t_\s) + \left[\frac{(t - t_\s - \Delta_\s)(y_{\s,\ii} - p_\ii(t_\s))}{\Delta} \right], t \in (t_\s + \Delta_\s, t_{\s+1}]
    \end{aligned}
\end{equation*}
for all $\ii\in\{1,\dots,n_{\mathrm{gfl}}\}$, where $y_{\s,\ii}$ denotes the solution of the following problem
\begin{equation}\label{eq:gfl_optimization}
    \begin{aligned}
    y_\s := & \argmin_{y_{\s,1},\dots, y_{\s,n_\mathrm{gfl}}} \sum_{\ii=1}^{n_\mathrm{gfl}} \ \frac{1}{2}(P_{y_{\ii}}^\star)^2 + b_\ii P_{y_{\ii}}^\star + c_\ii \\
   &\mbox{subject to} \  0 \leq y_{\s,\ii} \leq 1 \ \forall \ii \in \{1,\dots, n_\mathrm{gfl}\}, \\
   & \hspace{0.6in} \sum_{\ii=1}^{n_\mathrm{gfl}} P_{y_{\ii}}^\star = \textstyle \rho_\mathrm{d}(t_\s) \\
   &\hspace{-0.25in} P_{y_{\ii}}^\star := P_\ii^{\min} + y_{\s,\ii} (P_\ii^{\max} - P_\ii^{\min}) \ \forall \ii \in \{1,\dots, n_\mathrm{gfl}\},
\end{aligned}
\end{equation}
where $y_\s=[y_{\s,1},\dots,y_{\s,n_\mathrm{gfl}}]\in\mathbb{R}^{n_\mathrm{gfl}}$ denotes the solution of the optimization problem, and $P_\ii^{\min}$ and $P_\ii^{\max}$ are the active power limits of GFL-IIDG $\ii$, as in~\eqref{eq:Pi_star_design_law} and $\rho_\mathrm{d}(t_\s)$ is the total active power requirement. The GFL-IIDG power set-point estimate in~\eqref{eq:gfl_optimization} is parametrized as
$P_{y_\ii}^\star:=P_\ii^{\min}+y_{\s,\ii}(P_\ii^{\max}-P_\ii^{\min})$, with $0\leq y_{\s,\ii}\leq 1$, ensuring that the resulting active power set-points satisfy capacity constraints. The solution to~\eqref{eq:gfl_optimization}  determines the GFL-IIDG power reference $P_\ii^\star$ in~\eqref{eq:Pi_star_design_law}. Under the fast inner-current control loop, the GFL-IIDG output $P_\ii^\mathrm{gfl}$ tracks $P_\ii^\star$ with negligible, assumed zero, error~\cite{yazdani}. Therefore, 
\begin{align}\label{eq:gfl_power}
    \textstyle P_\ii^\mathrm{gfl} (t) = P_\ii^\star(t) = P_\ii^{\min} + p_\ii(t)(P_\ii^{\max} - P_\ii^{\min}) \ \forall t.
\end{align}
Problems~\eqref{eq:gfm_optimization}-\eqref{eq:gfl_optimization} can be written as
\begin{equation}\label{eq:opt_prob_equivalent}
    \begin{aligned}
	& \minimize_{\zeta_1, \zeta_2, \dots, \zeta_{n_\zeta}} \quad \sum_{\ii=1}^{n_\zeta} f_\ii(\zeta_\ii)\\
	& \mbox{subject to} \hspace{0.2in} \zeta_\ii = \zeta_\jj \ \forall \ii, \jj \in \{1,2,\dots,n_\zeta\}, \\
    & \hspace{0.75in} \zeta_\ii \in \mathcal{X}_\ii \  \forall \ii \in \{1,\dots,n_\zeta\}, \\ 
    & \hspace{0.75in} \overline{H}_\ii \zeta_\ii  = \overline{h}_\ii \  \forall \ii \in \{1,\dots,n_\zeta\},\\
    & \hspace{0.75in} \underline{H}_\ii \zeta_\ii \leq \underline{h}_\ii \  \forall \ii \in \{1,\dots,n_\zeta\},
\end{aligned}
\end{equation}
Here, each IIDG maintains a local decision variable $\zeta_\ii$. The objective $f_\ii:\mathbb{R}^{n_\zeta}\rightarrow\mathbb{R}$, constraint set $\mathcal{X}_\ii\subseteq\mathbb{R}^{n_\zeta}$, equality constraints $\overline{H}_\ii\zeta_\ii=\overline{h}_\ii$, and inequality constraints $\underline{H}_\ii\zeta_\ii\leq\underline{h}_\ii$ are local to IIDG $\ii$ and can be enforced in a decentralized manner. The consensus constraints $\zeta_\ii=\zeta_\jj$ couple the local decisions and require network-level coordination. We therefore solve~\eqref{eq:opt_prob_equivalent} using the distributed discrete-time $\admm$ algorithm developed in~\cite{ADMM_tac} that has a geometric rate of convergence. After $\theta$ iterations,
\begin{align}\label{eq:admm_iter_comp}
\|\zeta_\ii^{(\theta)} - \zeta^\star \|^2 \leq \Upsilon (0.75)^{\theta}, \quad \mbox{for all} \ \ii,
\end{align}
where $\Upsilon$ is a known constant determined from the problem data. Hence, accuracy $\epsilon$ requires only $\theta_\epsilon=O(\log(1/\epsilon))$ iterations, and the sampling interval $\Delta_\s$ in the sampled-data secondary controller can be chosen accordingly. 
\section{Stability Analysis}\label{sec:stability_analysis}
Using~\eqref{eq:w_droop},~\eqref{eq:v_droop_0}, GFM-IIDG $\ii$ has the closed-loop dynamics,
\begin{align}
    \dot{\delta}_\ii &= \omega_\ii, \quad \tau_{P_\ii}\dot{\omega_\mathrm{i}} = -(\omega_\mathrm{i} - \omega_\mathrm{nom}) - r_{\omega_\mathrm{i}}P_\mathrm{i}, \label{eq:freq_dyan}\\
    \tau_{Q_\ii}\dot{V_\ii} &= -(V_\ii - \Vnom) - r_{V_\ii}Q_\ii + U_\ii^\star +  \tau_{Q_\ii}\dot{U}_\ii^\star. \label{eq:V_dyan}
\end{align}

\vspace{-0.2in}
\subsection{Analysis of GFM-IIDG Voltage Dynamics Loop}
\noindent Substituting~\eqref{eq:U_i_design_law} in~\eqref{eq:V_dyan} we get,
\begin{align}\label{eq:V_dot_expanded}
    \tau_{Q_\ii}\dot{V_\ii} &= -(V_\ii - \Vnom) - r_{V_\ii}Q_\ii + U_\ii^\star +  \tau_{Q_\ii}\dot{U}_\ii^\star \nonumber  \\
    & = -(V_\ii - \Vnom) - r_{V_\ii}Q_\ii + u_\ii - \beta_{Q_\ii} r_{V_\ii} Q^\mathrm{avg}_\ii \\
    & \hspace{-0.15in} + \beta_{V_\ii} (\Vnom - V_\ii) + \tau_{Q_\ii}\dot{u}_\ii -  \tau_{Q_\ii}\beta_{Q_\ii} r_{V_\ii} \dot{Q}^\mathrm{avg}_\ii - \tau_{Q_\ii}\beta_{V_\ii} \dot{V}_\ii \nonumber. 
\end{align}
Define
\begin{align}\label{eq:beta_constants}
    \hspace{-0.1in} \tilde{\beta}_{Q_\ii} := (1 + \beta_{Q_\ii}), \ \tilde{\beta}_{V_\ii} := (1 + \beta_{V_\ii}), \ \tilde{\beta}_\ii := \tilde{\beta}_{Q_\ii}/\tilde{\beta}_{V_\ii} \ \forall \ii.
\end{align}
Substituting $\tau_{Q_\ii} \dot{Q}^\mathrm{avg}_\ii = Q_\ii - Q^\mathrm{avg}_\ii$ in~\eqref{eq:V_dot_expanded} we get for all $\ii$,
\begin{align}\label{eq:V_dot_expanded_intermediate}
    \tau_{Q_\ii}(1 + \beta_{V_i})\dot{V_\ii} &= -(1+\beta_{V_\ii})(V_\ii - \Vnom) - r_{V_\ii}Q_\ii \nonumber \\
    & \hspace{-0.4in} - \beta_{Q_\ii} r_{V_\ii} Q^\mathrm{avg}_\ii -  \beta_{Q_\ii} r_{V_\ii} (Q_\ii - Q^\mathrm{avg}_\ii) + u_\ii + \tau_{Q_\ii}\dot{u}_\ii \nonumber \\
    & \hspace{-0.8in} = -\tilde{\beta}_{V_\ii}(V_\ii - \Vnom) - \tilde{\beta}_{Q_\ii}r_{V_\ii}Q_\ii + u_\ii + \tau_{Q_\ii}\dot{u}_\ii,
\end{align}
Therefore, using~\eqref{eq:V_dot_expanded_intermediate} we have for all $\ii \in \{1,2,\dots,n_\mathrm{gfm}\}$,
\begin{align}\label{eq:V_dot_expanded_final}
    \dot{V_\ii} = \textstyle -\frac{1}{\tau_{Q_\ii}}(V_\ii - \Vnom) - \frac{\tilde{\beta}_\ii r_{V_\ii}}{\tau_{Q_\ii}}Q_\ii + \frac{1}{\tau_{Q_\ii} \tilde{\beta}_{V_\ii}}u_\ii + \frac{1}{\tilde{\beta}_{V_\ii}}\dot{u}_\ii.
\end{align}
Let $V:=[V_1,\dots,V_{n_\mathrm{gfm}}]^\top\in\mathbb{R}^{n_\mathrm{gfm}}$ and $Q^{\mathrm{avg}}(V)=Q^{\mathrm{avg}}:=[Q^{\mathrm{avg}}_1,\dots,Q^{\mathrm{avg}}_{n_\mathrm{gfm}}]^\top\in\mathbb{R}^{n_\mathrm{gfm}}$. We establish an invariance result for the closed-loop voltage dynamics around the nominal operating point $\mathbf{\Vnom}:=\Vnom \mathbf{1}_{n_\mathrm{gfm}}$, where $\mathbf{1}_{n_\mathrm{gfm}}$ is the vector with all entries equal to $1$, and $\Qnom:=Q(\mathbf{\Vnom})$. The result provides an explicit forward-invariant set whose size depends on the network parameters and controller design, revealing how intrinsic system dissipation and control-induced effects jointly govern closed-loop boundedness and performance.

\begin{theorem}\label{thm:centered_invariance}
Consider centered voltage and filtered reactive-power variables $\widetilde{V}(t) = V(t) - \mathbf \Vnom, \widetilde{q}(t) = Q^{\mathrm{avg}}(t) - \Qnom$. Given radii $R_V>0$, $R_q>0$, define the set of centered states $\Omega(R_V,R_q) := \{(\widetilde V,\widetilde q):\|\widetilde V\|_2\le R_V,\;\|\widetilde q\|_2\le R_q\}$. Assume $\tau_{Q_\ii}\ge 1, \tilde\beta_{V_\ii}>0$ for all $\ii$. If $V_\ii>0$ for all $\ii$, then there exist constants $\alpha > 0, c_\star > 0, \nu > 0$ and $\rho_\Omega > 0$, defined explicitly from the system parameters and the radii $(R_V, R_q)$, satisfying 
$c_\star < \alpha \rho_\Omega$ and the sub-level set $\mathcal S_{\rho_\Omega}
:=
\{(\widetilde V,\widetilde q):\Psi(\widetilde V,\widetilde q)\le \rho_\Omega\}$ is forward invariant where, $\Psi(\widetilde V,\widetilde q)
:= \frac16 \widetilde V^\top \text{diag}(1/\tau_{Q_\ii}) \widetilde V
+\frac{\nu}{2}\|\widetilde q\|_2^2$. If $(\widetilde{V}(0), \widetilde{q}(0)) \in \Omega(R_V,R_q)$
then $\|\widetilde V(t)\|_2\le R_V, \mbox{and} \ \|\widetilde q(t)\|_2\le R_q \ \forall t \geq 0.$
\end{theorem}
\begin{proof}
Define $u(t) := [u_1(t), \dots, u_{n_\mathrm{gfm}}(t)]^\top \in \mathbb{R}^{n_\mathrm{gfm}}$, and $\dot{u}(t) := [\dot{u}_1(t), \dots, \dot{u}_{n_\mathrm{gfm}}(t)]^\top \in \mathbb{R}^{n_\mathrm{gfm}}$, $\unom_\ii :=\tilde\beta_{V_\ii}\tilde\beta_\ii r_{V_\ii} \Qnom_\ii$, for all $\ii$, $\unom := [\unom_1, \dots, \unom_{n_{\mathrm{gfm}}}]^\top \in \mathbb{R}^{n_\mathrm{gfm}}$, $d_\ii := |B_{\ii \ii}| + \textstyle \sum_{k\in N_\ii}|B_{\ii \kk}|$ for all $\ii$, and $d_B := \Big(\sum_{\ii=1}^{n_\mathrm{gfm}} d_\ii^2\Big)^{1/2}$. Fix $R_V>0$, $R_q>0$, and define
\begin{align*}
\Omega(R_V,R_q) &:= \{(\widetilde V,\widetilde q):\|\widetilde V\|_2\le R_V,\;\|\widetilde q\|_2\le R_q\}, \\
L_Q &:= d_B \bigl(2\|\Vnom\|_2+R_V\bigr).
\end{align*}
Define
\begin{equation}\label{eq:constants}
\begin{aligned}
B_1 &= \max_\ii |1+\beta_{Q_\ii}|+\sqrt{n_{\mathrm{gfm}}}\max_\ii |\beta_{V_\ii}|, \\
B_2 &= \max_\ii |\beta_{V_\ii}r_{V_\ii}|+\sqrt{n_{\mathrm{gfm}}}\max_\ii |(1+\beta_{Q_\ii})r_{V_\ii}|, \\
B_0 &= \sqrt{n_{\mathrm{gfm}}}V_\Delta + B_2\|\Qnom\|_2 + \|\unom\|_2, \\
X_\star &:=  \textstyle B_0+B_1R_V+B_2R_q, \  
D_\star:=\frac{2X_\star}{\Delta}. 
\end{aligned}
\end{equation}
Let $ \tau_{\min}:=\min_\ii \tau_{Q_\ii}, \tau_{\max}:=\max_\ii \tau_{Q_\ii}, \tilde\beta_{V,\min}:=\min_\ii \tilde \beta_{V_\ii}, \boldsymbol{\tau}_{\beta r} := \text{diag}((\tilde{\beta}_\ii r_{V_\ii})/\tau_{Q_\ii}), \boldsymbol{\tau}_{\beta_V} := \text{diag}(1/(\tilde{\beta}_{V_\ii} \tau_{Q_\ii})), $ $ \boldsymbol{\beta}_V := \text{diag}(1/\tilde{\beta}_{V_\ii}), \boldsymbol{\tau} := \text{diag}(1/\tau_{Q_\ii}),$
and $K_u \le \frac{1}{3\tilde\beta_{V,\min}\tau_{\min}^2}, K_d\le \frac{1}{3\tilde\beta_{V,\min}\tau_{\min}}.$ Choose $\varepsilon_{r,1},\varepsilon_{r,2},\varepsilon_u,\varepsilon_d>0$, and 
$\varepsilon_q\in(0,2/\tau_{\max})$. Define $\varepsilon_r:=\varepsilon_{r,1}+\varepsilon_{r,2}$, and 
\begin{align*}
d_r(R_V) &:= \textstyle  \frac{1}{18\varepsilon_{r,1}}
\|\boldsymbol{\tau}_{\beta r}\boldsymbol{\tau} \Qnom \|_2^2  + 
\frac{L_Q^2}{18\varepsilon_{r,2}}
\|\boldsymbol{\tau}\boldsymbol{\tau}_{\beta r}\mathbf{\Vnom}\|_2^2, \\
 & \hspace{0.2in} \textstyle + \frac{1}{3} |
\mathbf{\Vnom}^\top
\boldsymbol{\tau}\boldsymbol{\tau}_{\beta r}\Qnom|, \\
c_V &:= \textstyle \frac{1}{3\tau_{\max}^2}-\frac{\varepsilon_r+\varepsilon_u+\varepsilon_d}{2}.
\end{align*}
$\varepsilon_{r,1},\varepsilon_{r,2},\varepsilon_u,\varepsilon_d>0$ be chosen such that $c_V > 0$. Further, let $\nu>0, \varepsilon_q > 0$ be such that
$0<\nu< \frac{2\varepsilon_q\tau_{\min}^2\,c_V}{L_Q^2}.$ Let
\begin{align*}
a_V &:= \textstyle c_V-\nu\frac{L_Q^2}{2\varepsilon_q\tau_{\min}^2}, \ a_q := \nu\Bigl(\frac{1}{\tau_{\max}}-\frac{\varepsilon_q}{2}\Bigr), \\
M_\Psi &:= \textstyle \max\left\{\frac{1}{6\tau_{\min}},\frac{\nu}{2}\right\}, \ 
m_\Psi:=\min\left\{\frac{1}{6\tau_{\max}},\frac{\nu}{2}\right\}, \\
\alpha &:= \textstyle \frac{\min\{a_V,a_q\}}{M_\Psi}, \
c_\star := \textstyle d_r(R_V)
+\frac{K_u^2}{2\varepsilon_u}X_\star^2
+\frac{K_d^2}{2\varepsilon_d}D_\star^2,\\
\rho_\Omega &:= \textstyle \min\left\{\frac{R_V^2}{6\tau_{\max}},
\frac{\nu R_q^2}{2}
\right\}.
\end{align*} 
We divide the proof into several steps.

\textbf{Step 1: Centered dynamics.}
Since
\[
V = \widetilde{V} + \mathbf{\Vnom},
\qquad
Q^{\mathrm{avg}}=\widetilde q + \Qnom,
\]
the voltage dynamics~\eqref{eq:V_dot_expanded} can be written as
\begin{align}
\dot{\widetilde V}
&=
-\boldsymbol\tau \widetilde V
-\boldsymbol\tau_{\beta r}\bigl(Q(V) - \Qnom\bigr) \nonumber \\
&\hspace{1in} + \boldsymbol\tau_{\beta_V}(u - \unom)
+\boldsymbol\beta_V \dot u. \label{eq:centered_voltage_dyn_proof}
\end{align}
Also, using $\tau_{Q_\ii}\dot Q_\ii^{\mathrm{avg}} = -Q_\ii^{\mathrm{avg}} + Q_\ii,$
\begin{equation}\label{eq:centered_filter_dyn_proof}
\dot{\widetilde q}
= -\boldsymbol{\tau} \widetilde q + \boldsymbol{\tau}\bigl(Q(V) - \Qnom \bigr).
\end{equation}

\textbf{Step 2: Bound the optimizer $x_\s$.}
Because of consensus constraints $x_{\s,\ii} = x_{\s,\jj}$, every feasible solution has the form $ x_\s = c_{x_\s} \mathbf{1}_{n_\mathrm{gfm}}$,
for some scalar $c_{x_\s} \in \mathbb R$. Therefore, $\|x\|_2 = \sqrt{n_\mathrm{gfm}} |c_{x_\s}|.$
Next, from~\eqref{eq:gfm_optimization},
\[
\alpha_\ii (t_\s)
=
(1+\beta_{Q_\ii})\Vnom - (1+\beta_{Q_\ii})V_\ii - \beta_{V_\ii}r_{V_\ii}Q_\ii^{\mathrm{avg}}(t_\s).
\]
Thus,
\begin{align*}
\|\alpha \|_2
\le &  \max_\ii |1+\beta_{Q_\ii}| \|\mathbf{\Vnom} - V\|_2 +  
\max_\ii |\beta_{V_\ii}r_{V_\ii}| \|Q^{\mathrm{avg}}\|_2.
\end{align*}
Let $L_\alpha^V := \max_\ii |1+\beta_{Q_\ii}|, 
L_\alpha^Q := \textstyle \max_\ii |\beta_{V_\ii}r_{V_\ii}|.$
Then, $\|\alpha\|_2
\le L_\alpha^V \|\mathbf{\Vnom} - V\|_2 + L_\alpha^Q \|Q^{\mathrm{avg}}\|_2.$
Define, $ \overline{\alpha} := \frac{1}{n_\mathrm{gfm}}\sum_{\ii=1}^{n_\mathrm{gfm}} \alpha_\ii$,
and by Cauchy-Schwarz inequality,
\begin{align*}
|\bar \alpha|
= \textstyle
\frac{1}{n_\mathrm{gfm}} |\mathbf{1}_{n_\mathrm{gfm}}^\top \alpha|
& \textstyle \le
\frac{1}{\sqrt{n_\mathrm{gfm}}}\|\alpha\|_2 \\
& \leq \textstyle \frac{L_\alpha^V}{\sqrt n_\mathrm{gfm}}\|\mathbf{\Vnom} - V\|_2
+
\frac{L_\alpha^Q}{\sqrt n_\mathrm{gfm}}\|Q^{\mathrm{avg}}\|_2.
\end{align*}
Now consider sets $D_\ii^{V_\Delta}$ in~\eqref{eq:gfm_optimization}, with the centers given by
\begin{align*}
m_\ii
& =
(1+\beta_{Q_\ii})r_{V_\ii}Q_\ii^{\mathrm{avg}}(t_\s)
- \beta_{V_\ii}(\Vnom - V_\ii (t_\s))
\end{align*}
Therefore,
\begin{align*}
|m_\ii| \le |\beta_{V_\ii}||\Vnom -  V_\ii (t_\s)| + |(1+\beta_{Q_\ii})r_{V_\ii}||Q_\ii^{\mathrm{avg}}(t_\s)|.
\end{align*}
Taking the maximum over $\ii$ yields
\[
\|m\|_\infty
\le m_V\|V(t_\s) - \Vnom\|_2 + m_Q\|Q^{\mathrm{avg}}(t_\s)\|_2,
\]
where
\begin{align*}
m_V := \max_\ii |\beta_{V_\ii}|, \ 
m_Q := \textstyle \max_\ii |(1+\beta_{Q_\ii})r_{V_\ii}|.
\end{align*}
Therefore, the solution $x_\s = c_{x_\s} \mathbf{1}_{n_\mathrm{gfm}}$ is bounded,
\[
|c_{x_\s}|
\le
|\bar\alpha| + \|m\|_\infty + V_\Delta.
\]
Substituting the previously derived bounds gives
\begin{align*}
|c_{x_\s}|
\leq & \textstyle 
\left(\frac{L_\alpha^V}{\sqrt{n_\mathrm{gfm}}} + m_V \right) \|V(t_\s) - \mathbf{\Vnom}\|_2
+ V_\Delta \\
& + \textstyle  \left(\frac{L_\alpha^Q}{\sqrt{n_\mathrm{gfm}}} + m_Q \right) \|Q^{\mathrm{avg}}(t_\s)\|_2
\end{align*}
Multiplying by $\sqrt{n_\mathrm{gfm}}$ and using $\|x_\s \|_2=\sqrt{n_\mathrm{gfm}} |c_{x_\s}|$, we conclude
\begin{align}\label{eq:bound_on_xs}
\|x_\s \|_2
\le B_1\|\tilde{V}(t_\s)\|_2 + B_2\|Q^{\mathrm{avg}}(t_\s)\|_2 + \sqrt{n_\mathrm{gfm}}V_\Delta,
\end{align}
where $B_1,B_2$ are exactly those given in~\eqref{eq:constants}. 
Hence, by the triangle inequality,
\begin{align}
\|x_\s - \unom\|_2
&\le
\|x_\s\|_2+\|\unom\|_2 \nonumber\\
& \le
\sqrt{n_{\mathrm{gfm}}}V_\Delta
+
B_1\|\widetilde V(t_\s)\|_2
+
B_2\|\widetilde q(t_\s)\|_2 \nonumber \\
& \hspace{0.2in} +
B_2\|\Qnom\|_2
+
\|\unom\|_2 \nonumber\\
&=
B_0 + B_1\|\widetilde V(t_\s)\|_2+B_2\|\widetilde q(t_\s)\|_2.
\label{eq:x_minus_unom_bound_proof}
\end{align}
Now assume $(\widetilde V(t_\s),\widetilde q(t_\s))\in \mathcal{S}_{\rho_\Omega}$, then
\[
\|\widetilde V(t_\s)\|_2\le R_V,
\qquad
\|\widetilde q(t_\s)\|_2\le R_q,
\]
and so~\eqref{eq:x_minus_unom_bound_proof} yields
\begin{equation}\label{eq:x_minus_unom_bound_region_proof}
\|x_\s - \unom\|_2
\le
B_0+B_1R_V+B_2R_q
=
X_\star.
\end{equation}

\textbf{Step 3: Explicit bounds on $u - \unom$ and $\dot{u}$.}
For $t\in[t_\s,t_\s + \Delta_\s)$, the interpolation law gives
\begin{align*}
u(t) = u(t_\s).
\end{align*}
And for $t\in[t_\s + \Delta_\s,t_{\s+1})$, the interpolation law gives
\begin{align*}
u(t)
= \textstyle \Bigl(1-\frac{t - t_\s - \Delta_\s}{\Delta}\Bigr)u(t_\s)
+ \frac{t - t_\s - \Delta_\s }{\Delta}x_\s.
\end{align*}
Subtracting $\unom$ for $t \in [t_\s, t_{\s+1})$, we have
$u(t) - \unom$
$= \begin{cases}
u(t_\s) - \unom \\
\big(1-\frac{t-t_\s-\Delta_\s}{\Delta}\big)(u(t_\s) - \unom)
+
\frac{t-t_\s-\Delta_\s}{\Delta}(x_s - \unom).
\end{cases} 
$ 

Since both coefficients are nonnegative and sum to one,
\begin{align}
\|u(t) - \unom\|_2
&\le \textstyle 
\Bigl(1-\frac{t-t_\s-\Delta_\s}{\Delta}\Bigr)\|u(t_s)- \unom\|_2 \nonumber 
\\
& \textstyle \hspace{0.2in}+ \frac{t-t_\s-\Delta_\s}{\Delta}\|x_\s - \unom\|_2 \nonumber\\
& \hspace{-0.4in}\le
\max \bigl\{\|u(t_\s) - \unom\|_2, \|x_\s - \unom\|_2 \bigr\}.\label{eq:u_convex_bound_proof}
\end{align}
Assume $\|u(0) - \unom\|_2\le X_\star$. Then, using~\eqref{eq:x_minus_unom_bound_region_proof}, we prove by induction over the sampling intervals $t \in [t_\s, t_{\s+1})$ that
\begin{equation}\label{eq:u_minus_unom_global_bound_proof}
\|u(t) - \unom\|_2\le X_\star.
\end{equation}
as long as $(\widetilde V(t),\widetilde q(t))\in \mathcal{S}_{\rho_\Omega}$. Indeed, it is true at $t=0$ by assumption. If it holds at $t=t_\s$, then~\eqref{eq:u_convex_bound_proof} and~\eqref{eq:x_minus_unom_bound_region_proof} imply it holds for all $t\in[t_\s,t_{\s+1})$. Also, $\dot u(t)= \begin{cases}
    0 & t\in[t_\s + \Delta_\s)\\
    \frac{x_\s - u(t_\s)}{\Delta} & t\in[t_\s + \Delta_\s,t_{\s+1})
\end{cases}$.

Hence,
\begin{align}
\|\dot u(t)\|_2
& \textstyle \le \frac{1}{\Delta}\bigl(\|x_\s - \unom\|_2+\|u(t_\s) -\unom\|_2\bigr) \nonumber\\
& \textstyle \le \frac{1}{\Delta}(X_\star+X_\star)
= \frac{2X_\star}{\Delta}
= D_\star. \label{eq:udot_bound_proof}
\end{align}
Therefore, on $\Omega(R_V,R_q)$,
\begin{equation}\label{eq:u_udot_explicit_bounds_proof}
\|u(t) - \unom \|_2\le X_\star,
\quad
\|\dot u(t)\|_2 \le D_\star.
\end{equation}

\textbf{Step 4: Local Lipschitz bound on $Q(V) - \Qnom$.}
From the quadratic structure of the reactive-power map,
\[
\|Q(V) - Q(W)\|_2
\le
d_B\bigl(\|V\|_2+\|W\|_2\bigr)\|V-W\|_2.
\]
Apply this with $W = \mathbf{\Vnom}$. Then
\[
\|Q(V) - \Qnom\|_2
\le
d_B\bigl(\|V\|_2+\|\mathbf \Vnom\|_2\bigr)\|V - \Vnom\|_2.
\]
Since $V=\widetilde V+\mathbf \Vnom, \ 
\|V\|_2\le \|\widetilde V\|_2+\|\mathbf \Vnom\|_2,
$
and on $\Omega(R_V,R_q)$ we have $\|\widetilde V\|_2\le R_V$, it follows that
\begin{align*}
\|V\|_2\le R_V+\|\mathbf \Vnom\|_2.
\end{align*}
Hence
\begin{align}
\|Q(V) - \Qnom\|_2
&\le
d_B\bigl(R_V+\|\mathbf \Vnom\|_2+\|\mathbf \Vnom\|_2\bigr)\|\widetilde V\|_2 \nonumber\\
&\hspace{-0.8in} =
d_B\bigl(2\|\mathbf \Vnom\|_2+R_V\bigr)\|\widetilde V\|_2 = L_Q\|\widetilde V\|_2.
\label{eq:LQ_bound_proof}
\end{align}

\textbf{Step 5: Derivative of the voltage Lyapunov function.}
Define
\[
\textstyle \Phi(\widetilde V):=\frac16 \widetilde V^\top \boldsymbol\tau \widetilde V.
\]
Then
\[
\textstyle \dot\Phi
=
\dot{\widetilde{V}}^\top \frac16 \boldsymbol\tau \widetilde V
+
\widetilde V^\top \frac16 \boldsymbol\tau \dot{\widetilde V}
= \frac13 \widetilde V^\top \boldsymbol\tau \dot{\widetilde V},
\]
since $\boldsymbol\tau$ is diagonal and symmetric. Using~\eqref{eq:centered_voltage_dyn_proof},
\begin{align}
\dot\Phi
&= \textstyle -\frac13 \widetilde V^\top \boldsymbol\tau^2 \widetilde V
-\frac13 \widetilde V^\top \boldsymbol\tau\boldsymbol\tau_{\beta r}(Q - \Qnom) \nonumber\\
&\hspace{0.3in} \textstyle 
+\frac13 \widetilde V^\top \boldsymbol\tau\boldsymbol\tau_{\beta_V}(u - \unom)
+\frac13 \widetilde V^\top \boldsymbol\tau\boldsymbol\beta_V\dot u.
\label{eq:Phi_dot_centered_intermediate_proof}
\end{align}
To make the reactive term symmetric, write
\begin{align}
\dot\Phi
&= \textstyle
-\frac13 \widetilde V^\top \boldsymbol\tau^2 \widetilde V
+\mathcal{T}_Q
+\frac13 \widetilde V^\top \boldsymbol\tau\boldsymbol\tau_{\beta_V}(u - \unom) \nonumber \\
& \textstyle \hspace{0.4in} +\frac13 \widetilde V^\top \boldsymbol\tau\boldsymbol\beta_V \dot u,
\label{eq:Phi_dot_centered_with_TQ_proof}
\end{align}
where
\begin{align}
\mathcal{T}_Q
 & \textstyle := -\frac13 \widetilde V^\top \boldsymbol\tau\boldsymbol\tau_{\beta r}(Q - \Qnom) \nonumber \\
 & \hspace{-0.25in} \textstyle = 
-\frac16 \widetilde V^\top \boldsymbol\tau\boldsymbol\tau_{\beta r}(Q-\Qnom)
-\frac16 (Q - \Qnom)^\top \boldsymbol\tau_{\beta r}\boldsymbol\tau \widetilde V. \label{eq:TQ_def_proof}
\end{align}

We now estimate $\mathcal{T}_Q$ using the improved centered decomposition. Since $\widetilde V = V - \mathbf \Vnom$,
we expand~\eqref{eq:TQ_def_proof} as
\begin{align}
\mathcal{T}_Q
&= \textstyle 
-\frac16 (V - \mathbf \Vnom)^\top \boldsymbol\tau\boldsymbol\tau_{\beta r}(Q - \Qnom) \nonumber\\
&\hspace{0.2in}
\textstyle -\frac16 (Q - \Qnom)^\top \boldsymbol\tau_{\beta r}\boldsymbol\tau (V - \mathbf \Vnom) \nonumber\\
&= \textstyle 
-\frac16 V^\top \boldsymbol\tau\boldsymbol\tau_{\beta r}Q
-\frac16 Q^\top \boldsymbol\tau_{\beta r}\boldsymbol\tau V \nonumber\\
&\hspace{0.2in} \textstyle 
+\frac16 V^\top \boldsymbol\tau\boldsymbol\tau_{\beta r} \Qnom
+\frac16 (\Qnom)^\top \boldsymbol\tau_{\beta r}\boldsymbol\tau V \nonumber\\
&\hspace{0.2in} \textstyle 
+\frac16 (\mathbf \Vnom)^\top \boldsymbol\tau\boldsymbol\tau_{\beta r}Q
+\frac16 Q^\top \boldsymbol\tau_{\beta r}\boldsymbol\tau \mathbf \Vnom \nonumber\\
&\hspace{-0.2in} \textstyle 
-\frac16 (\mathbf \Vnom)^\top \boldsymbol\tau\boldsymbol\tau_{\beta r}\Qnom
-\frac16 (\Qnom)^\top \boldsymbol\tau_{\beta r}\boldsymbol\tau \mathbf \Vnom. \label{eq:TQ_expanded_full_proof}
\end{align}
Rearranging terms,
\begin{align}
\mathcal{T}_Q
&= \textstyle 
-\frac16 V^\top \boldsymbol\tau\boldsymbol\tau_{\beta r}Q
-\frac16 Q^\top \boldsymbol\tau_{\beta r}\boldsymbol\tau V +\frac13 (\Qnom)^\top \boldsymbol\tau_{\beta r}\boldsymbol\tau \widetilde V \nonumber\\
& \textstyle \hspace{0.2in}
+\frac13 (\mathbf \Vnom)^\top \boldsymbol\tau\boldsymbol\tau_{\beta r}Q.
\label{eq:TQ_rearranged_proof}
\end{align}
Using $Q(V)=\mathrm{diag}(V)\mathcal B V$,
and defining $\Sigma(V):=\mathrm{diag}\!\left(\frac{\tilde\beta_\ii r_{V_\ii}V_\ii}{2\tau_{Q_\ii}^2}\right)$,
the first two terms in~\eqref{eq:TQ_rearranged_proof} become
\[ \textstyle 
-\frac13 V^\top \bigl(\Sigma(V)\mathcal B+\mathcal B^\top \Sigma(V)\bigr)V.
\]
Note that the matrix $\Sigma(V) \mathcal{B} + \mathcal{B}^\top \Sigma(V)$ is positive semi-definite (see Lemma~\ref{lem:psd_proof}),
and therefore
\[
\textstyle -\frac13 V^\top \bigl(\Sigma(V)\mathcal B+\mathcal B^\top \Sigma(V)\bigr)V \le 0.
\]
Hence,
\begin{align}
\mathcal{T}_Q
& \textstyle \le
\frac13 (\Qnom)^\top \boldsymbol\tau_{\beta r}\boldsymbol\tau \widetilde V
+\frac13 (\mathbf \Vnom)^\top \boldsymbol\tau\boldsymbol\tau_{\beta r}Q. \label{eq:TQ_after_PSD_drop_proof}
\end{align}
We now bound the two remaining terms. For the first term, Young's inequality gives
\begin{align}
\left| \textstyle 
\frac13 (\Qnom)^\top \boldsymbol\tau_{\beta r}\boldsymbol\tau \widetilde V
\right|
& \textstyle \le
\frac{\varepsilon_{r,1}}{2}\|\widetilde V\|_2^2
+
\frac{1}{18\varepsilon_{r,1}}
\|\boldsymbol\tau_{\beta r}\boldsymbol\tau \Qnom\|_2^2.
\label{eq:TQ_term1_bound_proof}
\end{align}
For the second term, using~\eqref{eq:LQ_bound_proof},
\begin{align}
& \left| \textstyle
\frac13 (\mathbf \Vnom)^\top \boldsymbol\tau\boldsymbol\tau_{\beta r}Q
\right| \nonumber \\
& \textstyle \le
\frac13
\|\boldsymbol\tau\boldsymbol\tau_{\beta r}\mathbf \Vnom\|_2
\|Q - \Qnom\|_2 + \left| \textstyle
\frac{1}{3} \mathbf \Vnom^\top \boldsymbol\tau\boldsymbol\tau_{\beta r}\Qnom
\right| \nonumber \\
& \textstyle \le
\frac13
\|\boldsymbol\tau\boldsymbol\tau_{\beta r}\mathbf \Vnom\|_2
L_Q\|\widetilde V\|_2 + \left| \textstyle
\frac{1}{3} \mathbf \Vnom^\top \boldsymbol\tau\boldsymbol\tau_{\beta r}\Qnom
\right| \nonumber\\
& \textstyle \le
\frac{\varepsilon_{r,2}}{2}\|\widetilde V\|_2^2
+
\frac{L_Q^2}{18\varepsilon_{r,2}}
\|\boldsymbol\tau\boldsymbol\tau_{\beta r}\mathbf \Vnom\|_2^2 + \left| \textstyle
\frac{1}{3} \mathbf \Vnom^\top \boldsymbol\tau\boldsymbol\tau_{\beta r}\Qnom
\right|.
\label{eq:TQ_term2_bound_proof}
\end{align}
Combining~\eqref{eq:TQ_after_PSD_drop_proof},~\eqref{eq:TQ_term1_bound_proof}, and~\eqref{eq:TQ_term2_bound_proof}, we get
\begin{align}
\mathcal{T}_Q
& \le \textstyle \frac{\varepsilon_r}{2}\|\widetilde V\|_2^2+d_r(R_V),
\label{eq:TQ_final_bound_proof}
\end{align}
where
\[
\varepsilon_r=\varepsilon_{r,1}+\varepsilon_{r,2}.
\]
Next, the dissipative quadratic term satisfies
\begin{align}
\textstyle -\frac13 \widetilde V^\top \boldsymbol\tau^2 \widetilde V
& \textstyle \le
-\frac{1}{3\tau_{\max}^2}\|\widetilde V\|_2^2,
\label{eq:tau2_dissipation_proof}
\end{align}
because the smallest eigenvalue of $\boldsymbol\tau^2$ is $1/\tau_{\max}^2$.

For the control term involving $u - \unom$, using~\eqref{eq:u_udot_explicit_bounds_proof} and Young's inequality,
\begin{align}
\textstyle \frac13 \widetilde V^\top \boldsymbol\tau\boldsymbol\tau_{\beta_V}(u - \unom)
&\le
K_u\|\widetilde V\|_2\|u - \unom\|_2 \nonumber\\
& \textstyle \le
\frac{\varepsilon_u}{2}\|\widetilde V\|_2^2
+
\frac{K_u^2}{2\varepsilon_u}\|u - \unom\|_2^2 \nonumber\\
& \textstyle \le
\frac{\varepsilon_u}{2}\|\widetilde V\|_2^2
+
\frac{K_u^2}{2\varepsilon_u}X_\star^2.
\label{eq:u_term_bound_proof}
\end{align}
Similarly, for the $\dot u$ term,
\begin{align}
\textstyle \frac13 \widetilde V^\top \boldsymbol\tau\boldsymbol\beta_V\dot u
& \textstyle \le
K_d\|\widetilde V\|_2\|\dot u\|_2 \nonumber\\
& \textstyle \le
\frac{\varepsilon_d}{2}\|\widetilde V\|_2^2
+
\frac{K_d^2}{2\varepsilon_d}\|\dot u\|_2^2 \nonumber\\
&\textstyle \le
\frac{\varepsilon_d}{2}\|\widetilde V\|_2^2
+
\frac{K_d^2}{2\varepsilon_d}D_\star^2.
\label{eq:udot_term_bound_proof}
\end{align}
Substituting~\eqref{eq:TQ_final_bound_proof}, \eqref{eq:tau2_dissipation_proof}, \eqref{eq:u_term_bound_proof}, and \eqref{eq:udot_term_bound_proof} into~\eqref{eq:Phi_dot_centered_with_TQ_proof} yields
\begin{align}
\dot\Phi
&\textstyle \le
-\left(
\frac{1}{3\tau_{\max}^2}
-\frac{\varepsilon_r+\varepsilon_u+\varepsilon_d}{2}
\right)\|\widetilde V\|_2^2
+d_r(R_V) \nonumber\\
&\hspace{0.2in}
\textstyle +
\frac{K_u^2}{2\varepsilon_u}X_\star^2
+
\frac{K_d^2}{2\varepsilon_d}D_\star^2 \nonumber\\
&=
\textstyle -c_V\|\widetilde V\|_2^2
+d_r(R_V)
+
\frac{K_u^2}{2\varepsilon_u}X_\star^2
+
\frac{K_d^2}{2\varepsilon_d}D_\star^2.
\label{eq:Phi_dot_final_proof}
\end{align}

\textbf{Step 6: Derivative of the filter-energy term.}
From~\eqref{eq:centered_filter_dyn_proof},
\begin{align}
\textstyle \frac12\frac{d}{dt}\|\widetilde q\|_2^2
& = 
\textstyle \widetilde q^\top \dot{\widetilde q} =
-\widetilde q^\top \boldsymbol{\tau}\widetilde q
+
\widetilde q^\top \boldsymbol{\tau}(Q-\Qnom).
\label{eq:q_energy_derivative_proof}
\end{align}
Since
\[
\textstyle \boldsymbol{\tau} \succeq \frac{1}{\tau_{\max}}I,
\qquad
\|\boldsymbol{\tau}\|\le \frac{1}{\tau_{\min}},
\]
we obtain
\begin{align}
\textstyle \frac12\frac{d}{dt}\|\widetilde q\|_2^2
& \textstyle \le
-\frac{1}{\tau_{\max}}\|\widetilde q\|_2^2
+
\frac{1}{\tau_{\min}}\|\widetilde q\|_2\|Q - \Qnom\|_2.
\label{eq:q_energy_derivative_bound1_proof}
\end{align}
Using~\eqref{eq:LQ_bound_proof},
\begin{align}
\textstyle \frac12\frac{d}{dt}\|\widetilde q\|_2^2
&\le
\textstyle -\frac{1}{\tau_{\max}}\|\widetilde q\|_2^2
+
\frac{L_Q}{\tau_{\min}}\|\widetilde q\|_2\|\widetilde V\|_2.
\label{eq:q_energy_derivative_bound2_proof}
\end{align}
Applying Young's inequality with parameter $\varepsilon_q$,
\begin{align}
\textstyle \frac{L_Q}{\tau_{\min}}\|\widetilde q\|_2\|\widetilde V\|_2
& \textstyle \le \frac{\varepsilon_q}{2}\|\widetilde q\|_2^2
+
\frac{L_Q^2}{2\varepsilon_q\tau_{\min}^2}\|\widetilde V\|_2^2.
\label{eq:q_young_bound_proof}
\end{align}
Thus,
\begin{align}
\textstyle \frac12\frac{d}{dt}\|\widetilde q\|_2^2
& \textstyle \le
-\Bigl(\frac{1}{\tau_{\max}}-\frac{\varepsilon_q}{2}\Bigr)\|\widetilde q\|_2^2
+
\frac{L_Q^2}{2\varepsilon_q\tau_{\min}^2}\|\widetilde V\|_2^2.
\label{eq:q_energy_final_proof}
\end{align}

\textbf{Step 7: Derivative of the composite Lyapunov function and invariance.}
Define
\[
\Psi(\widetilde V,\widetilde q)
=
\Phi(\widetilde V)+\frac{\nu}{2}\|\widetilde q\|_2^2.
\]
Multiplying~\eqref{eq:q_energy_final_proof} by $\nu$ and adding to~\eqref{eq:Phi_dot_final_proof}, we get
\begin{align}
\dot\Psi
& \textstyle \le
-\left(
c_V-\nu\frac{L_Q^2}{2\varepsilon_q\tau_{\min}^2}
\right)\|\widetilde V\|_2^2 \nonumber\\
&\hspace{0.15in} \textstyle 
-\nu\Bigl(\frac{1}{\tau_{\max}}-\frac{\varepsilon_q}{2}\Bigr)\|\widetilde q\|_2^2
+
c_\star,
\label{eq:Psi_dot_intermediate_proof}
\end{align}
where
\[
\textstyle c_\star=
d_r(R_V)
+\frac{K_u^2}{2\varepsilon_u}X_\star^2
+\frac{K_d^2}{2\varepsilon_d}D_\star^2.
\]
Due to the choice of $\nu$,
\[
\textstyle a_V=
c_V-\nu\frac{L_Q^2}{2\varepsilon_q\tau_{\min}^2}>0.
\]
Also, since $\varepsilon_q<2/\tau_{\max}$,
\[
\textstyle a_q=
\nu\Bigl(\frac{1}{\tau_{\max}}-\frac{\varepsilon_q}{2}\Bigr)>0.
\]
Therefore,
\begin{align}
\dot\Psi
&\le
-a_V\|\widetilde V\|_2^2-a_q\|\widetilde q\|_2^2+c_\star.
\label{eq:Psi_dot_with_aVaQ_proof}
\end{align}
Next, by definition of $M_\Psi$,
\[
\Psi(\widetilde V,\widetilde q)
\le
M_\Psi\bigl(\|\widetilde V\|_2^2+\|\widetilde q\|_2^2\bigr).
\]
Hence,
\[
\textstyle a_V\|\widetilde V\|_2^2+a_q\|\widetilde q\|_2^2
\ge
\frac{\min\{a_V,a_q\}}{M_\Psi}\Psi
=
\alpha \Psi.
\]
Substituting this into~\eqref{eq:Psi_dot_with_aVaQ_proof}, we obtain
\begin{align}
\dot\Psi
&\le
-\alpha\Psi+c_\star.
\label{eq:Psi_dot_scalar_proof}
\end{align}
Since, $c_\star < \alpha \rho_\Omega$, 
whenever $\Psi=\rho_\Omega$,
\[
\dot\Psi\le -\alpha\rho_\Omega+c_\star<0.
\]
The estimates leading to
\eqref{eq:Psi_dot_scalar_proof} are valid whenever
\[
(\widetilde V(t),\widetilde q(t))\in \Omega(R_V,R_q),
\]
and for almost all $t$, since $u(t)$ is piecewise affine. Define the
first exit time
\[
T^\star
:=
\inf\left\{
t\ge 0:
(\widetilde V(t),\widetilde q(t))\notin \Omega(R_V,R_q)
\right\}.
\]
If no such time exists, then $T^\star=\infty$. On the interval $[0,T^\star)$, the trajectory belongs to
$\Omega(R_V,R_q)$. Therefore, all bounds derived above apply, and
\eqref{eq:Psi_dot_scalar_proof} holds for almost all
$t\in[0,T^\star)$:
\[
\dot\Psi(t)\le -\alpha\Psi(t)+c_\star .
\]
By the comparison lemma, for all $t\in[0,T^\star)$,
\[
\Psi(t)
\le
e^{-\alpha t}\Psi(0)
+
\frac{c_\star}{\alpha}\bigl(1-e^{-\alpha t}\bigr).
\]
If $\Psi(0)\le \rho_\Omega$, then using $c_\star<\alpha\rho_\Omega.$
Then, for every $t\in[0,T^\star)$,
\[
\Psi(t)
\le
e^{-\alpha t}\rho_\Omega
+
\frac{c_\star}{\alpha}\bigl(1-e^{-\alpha t}\bigr)
\le
\rho_\Omega .
\]
Thus,
\[
\Psi(t)\le \rho_\Omega,
\qquad
\forall t\in[0,T^\star).
\]
Next, by the definition of $\Psi$,
\[
\Psi(\widetilde V,\widetilde q)
= \textstyle 
\frac{1}{6} \widetilde V^\top \boldsymbol\tau \widetilde V
+
\frac{\nu}{2}\|\widetilde q\|_2^2.
\]
Finally, by the definition of $m_\Psi$,
\[
\textstyle \Psi(\widetilde V,\widetilde q)
\ge
\frac{1}{6\tau_{\max}}\|\widetilde V\|_2^2,
\qquad
\Psi(\widetilde V,\widetilde q)
\ge
\frac{\nu}{2}\|\widetilde q\|_2^2.
\]
Therefore, for all $t\in[0,T^\star)$,
\[
\|\widetilde V(t)\|_2^2
\le
6\tau_{\max}\Psi(t)
\le
6\tau_{\max}\rho_\Omega
\le
R_V^2,
\]
and
\[
\textstyle \|\widetilde q(t)\|_2^2
\le
\frac{2}{\nu}\Psi(t)
\le
\frac{2}{\nu}\rho_\Omega
\le
R_q^2.
\]
Hence,
\[
(\widetilde V(t),\widetilde q(t))\in \Omega(R_V,R_q),
\qquad
\forall t\in[0,T^\star).
\]
Suppose, for contradiction, that $T^\star<\infty$. Since the state trajectory is continuous, taking the limit $t \to T^\star$ gives
\[
\|\widetilde V(T^\star)\|_2\le R_V,
\qquad
\|\widetilde q(T^\star)\|_2\le R_q.
\]
Thus
\[
(\widetilde V(T^\star),\widetilde q(T^\star))
\in \Omega(R_V,R_q),
\]
which contradicts the definition of $T^\star$ as the first exit time
from $\Omega(R_V,R_q)$. Therefore, $T^\star=\infty$. Consequently,
\[
\Psi(t)\le \rho_\Omega,
\qquad
\forall t\ge 0,
\]
and hence
\[
\|\widetilde V(t)\|_2\le R_V,
\qquad
\|\widetilde q(t)\|_2\le R_q,
\qquad
\forall t\ge 0.
\]
Therefore, $\Omega(R_V,R_q)$ is forward invariant for all trajectories
starting in the sublevel set
\[
\mathcal S_{\rho_\Omega}
:=
\{(\widetilde V,\widetilde q):\Psi(\widetilde V,\widetilde q)\le \rho_\Omega\}.
\]
Moreover,
\[
\mathcal S_{\rho_\Omega}\subseteq \Omega(R_V,R_q),
\]
and $\mathcal S_{\rho_\Omega}$ is forward invariant.
This completes the proof.
\end{proof}

We next establish the existence of a steady state within the invariant region. Since the invariance confines trajectories to a compact set, fixed-point arguments can be used to characterize equilibrium behavior. We show that the closed-loop voltage dynamics admit a steady-state operating point consistent with the network power flow and controller structure, and that this equilibrium lies within the certified invariant region.

\begin{theorem}\label{thm:steady_state_existence}
Assume the hypotheses of Theorem~\ref{thm:centered_invariance} hold, and let $R_V<\Vnom$ be the radius therein. Let $\mathcal{K}_V
:=
\prod_{\ii=1}^{n_{\mathrm{gfm}}}[\Vnom-R_V, \ \Vnom+R_V]
\subset \mathbb R^{n_{\mathrm{gfm}}}$. Define the secondary-level voltage control map~\eqref{eq:voltage_reference} component-wise, for all $\ii=1,\dots,n_{\mathrm{gfm}}$, as $T_\ii(V)
:= \textstyle 
\Vnom
+ 
\frac{1}{(1+\beta_{V_\ii})}x_{\s,\ii}(V,Q(V)) - \frac{r_{V_\ii}(1 + \beta_{Q_\ii})}{(1+\beta_{V_\ii})}Q_\ii(V).$ 
Then,
i) $T$ is continuous on $\mathcal K_V$, and $T(\mathcal K_V)\subseteq \mathcal K_V$, ii) there exists $V^\star\in\mathcal K_V$ such that $T(V^\star)=V^\star$. Let $Q^\star := Q(V^\star),
\ (Q^{\mathrm{avg}})^\star:=Q^\star,$ and $u^\star:= x_\s(V^\star,Q^\star)$,
then the signals $V(t)\equiv V^\star,\
Q(t)\equiv Q^\star,\
Q^{\mathrm{avg}}(t)\equiv Q^\star,\
u(t)\equiv u^\star $ 
form a positive steady state of the voltage/filter/controller subsystem. 
\end{theorem}
\begin{proof}
Define $\mathcal K_V
:=
\prod_{i=1}^{n_{\mathrm{gfm}}}[\Vnom-R_V, \ \Vnom+R_V]
\subset \mathbb R^{n_{\mathrm{gfm}}}$, $d_{\max} := \max_\ii \textstyle  \left(|B_{\ii\ii}|+\sum_{k\in N_\ii}|B_{\ii\kk}|\right), M_Q := d_{\max}(\Vnom $ $ + R_V)^2, M_u := B_1\sqrt{n_{\mathrm{gfm}}} R_{V} + B_2\sqrt{n_{\mathrm{gfm}}}M_Q + \sqrt{n_\mathrm{gfm}}V_\Delta$, where, $B_1,B_2$ are as in~\eqref{eq:constants}
$\tilde\beta_{Q,\max}:=\max_\ii |\tilde\beta_{Q_\ii}|, r_{V,\max}:=\max_\ii |r_{V_\ii}|$. Let
\begin{equation}\label{eq:selfmap_condition}
\frac{M_u+\tilde\beta_{Q,\max}r_{V,\max}M_Q}{\tilde\beta_{V,\min}}
\le R_V.
\end{equation}

We proceed in several steps.

\textbf{Step 1: $\mathcal K_V$ is a nonempty compact convex positive set.}
Since $R_{V}>0$, the set $\mathcal K_V$ is nonempty. Because it is a product of closed bounded intervals, it is compact; because it is a box, it is convex. Since $\Vnom > R_{V} >0$. 
Hence, every $V\in\mathcal K_V$ is componentwise positive:
\[
V_\ii\in [\Vnom-R_{V},\Vnom+R_{V}]
\subset (0,\infty),
\quad \ii=1,\dots,n_{\mathrm{gfm}}.
\]

\textbf{Step 2: Map $T$ is continuous.}
By definition~\eqref{eq:reactivepower}, $Q(V)$ is continuous for any $V$ and the solution $x_\s$ of~\eqref{eq:gfm_optimization} is continuous in $(V,Q(V))$ (see Lemma~\ref{lem:continuous_sol}). Therefore, the composition
\[
V \mapsto x_\s\bigl(V,Q(V)\bigr)
\]
is continuous. Since each component $T_\ii$ is obtained from continuous operations, $T$ is continuous.

\textbf{Step 3: Uniform bounds for $Q(V)$ on $\mathcal K_V$.}
Let $V\in \mathcal K_V$. Since each component satisfies
\[
|V_\ii|\le \Vnom+R_{V},
\]
we have
\[
\|V\|_\infty\le \Vnom+R_{V}.
\]
For each $\ii$, using the reactive-power relation~\eqref{eq:reactivepower}, we obtain
\begin{align*}
    |Q_\ii(V)|
    &\le \textstyle 
    \left(|B_{\ii\ii}|+\sum_{\kk\in N_\ii}|B_{\ii\kk}|\right)\|V\|_\infty^2 \\
    &\le d_{\max}(\Vnom+R_{V})^2 = M_Q.
\end{align*}
Hence
\begin{equation}\label{eq:Qi_bound_proof}
|Q_\ii(V)|\le M_Q,
\quad \ii=1,\dots,n_{\mathrm{gfm}}.
\end{equation}
which yields $\|Q(V)\|_2
\le
\sqrt{n_{\mathrm{gfm}}}M_Q.$
Using~\eqref{eq:bound_on_xs}, we have
\begin{align*}\label{eq:bound_on_xs}
\|x_\s \|_2
&\le B_1\|\tilde{V}\|_2 + B_2\|Q^{\mathrm{avg}}\|_2 + \sqrt{n_\mathrm{gfm}}V_\Delta \\
& \hspace{-0.1in} \le \left(B_1\sqrt{n_{\mathrm{gfm}}} R_{V} + B_2\sqrt{n_{\mathrm{gfm}}}M_Q + \sqrt{n_\mathrm{gfm}}V_\Delta \right):= M_u.
\end{align*}


\textbf{Step 4: $T(\mathcal K_V)\subseteq \mathcal K_V$.}
Fix $V\in\mathcal K_V$. For each $\ii$, using~\eqref{eq:selfmap_condition}, and~\eqref{eq:Qi_bound_proof}, we get
\begin{align*}
|T_\ii(V)-\Vnom|
&= \textstyle \left|
\frac{x_{\s,\ii}(V,Q(V))-\tilde\beta_{Q_\ii}r_{V_\ii}Q_\ii(V)}{\tilde\beta_{V_\ii}}
\right| \\
& \textstyle \le
\frac{
|x_{\s,\ii}(V,Q(V))|
+
|\tilde\beta_{Q_\ii}|\,
|r_{V_\ii}|\,
|Q_\ii(V)|
}{\tilde\beta_{V_\ii}} \\
&\textstyle \le
\frac{M_u+\tilde\beta_{Q,\max}r_{V,\max}M_Q}{\tilde\beta_{V,\min}}
\le
R_{V},
\end{align*}
where the last inequality is exactly~\eqref{eq:selfmap_condition}. Therefore
\[
\Vnom-R_{V}\le T_\ii(V)\le \Vnom+R_{V},
\quad \ii=1,\dots,n_{\mathrm{gfm}},
\]
which shows that $T(V)\in \mathcal K_V$. Since $V\in\mathcal K_V$ was arbitrary, we conclude
\[
T(\mathcal K_V)\subseteq \mathcal K_V.
\]

\textbf{Step 5: Existence of a fixed point.}
The set $\mathcal K_V$ is nonempty, compact, and convex, and since
$T:\mathcal K_V\to\mathcal K_V$ is continuous, Brouwer's fixed-point theorem~\cite{brouwer1911abbildung} implies that there exists $V^\star\in\mathcal K_V$ such that
$ V^\star\in \mathcal K_V $ such that $ T(V^\star)=V^\star$.

\textbf{Step 6: Construction of the steady state.}
Define
\begin{align*}
Q^\star:=Q(V^\star),
\quad
(Q^{\mathrm{avg}})^\star:=Q^\star, \quad
u^\star:=x_\s(V^\star,Q^\star).
\end{align*}
We claim that the constant signals
\begin{align*}
V(t) \equiv V^\star, \
Q(t)\equiv Q^\star, \ Q^{\mathrm{avg}}(t) \equiv Q^\star, \ u(t)\equiv u^\star,
\end{align*}
solve the voltage/filter/controller subsystem.

First, since $u(t) = u^\star = x_\s(V^\star, Q^\star)  = c^\star \mathbf 1_{\mathrm{gfm}}$, for some constant $c^\star$, for all $t$, we have
\begin{align*}
\dot u_\ii(t)=0, \quad \ii=1,\dots,n_{\mathrm{gfm}}.
\end{align*}

Second, since $(Q^{\mathrm{avg}})^\star=Q^\star$, the averaging filter
\[
\tau_{Q_\ii}\dot{(Q^{\mathrm{avg}})}^\star = Q^\star_\ii - (Q^{\mathrm{avg}})^\star
\]
yields
\[
\dot{(Q^{\mathrm{avg}})}_\ii^\star(t)=0,
\quad \ii=1,\dots,n_{\mathrm{gfm}}.
\]

It remains to verify the voltage equation. The fixed-point identity $T(V^\star)=V^\star$ means that, for each $\ii$,
\begin{align*}
 V_\ii^\star
= \Vnom +
\frac{u_\ii^\star - \tilde\beta_{Q_\ii}r_{V_\ii}Q_\ii^\star}{\tilde\beta_{V_\ii}}.
\end{align*}
Equivalently,
\[
u_\ii^\star
=
\tilde\beta_{V_\ii}(V_\ii^\star - \Vnom)+\tilde\beta_{Q_\ii}r_{V_\ii}Q_\ii^\star.
\]
Using the relation $\tilde\beta_\ii=\tilde\beta_{Q_\ii}/\tilde\beta_{V_\ii}$, we may rewrite this as
\[ 0
= -\frac{1}{\tau_{Q_\ii}}(V_\ii^\star-\Vnom)
-\frac{\tilde\beta_\ii r_{V_\ii}}{\tau_{Q_\ii}}Q_\ii^\star
+\frac{1}{\tau_{Q_\ii}\tilde\beta_{V_\ii}}u_\ii^\star.
\]
Since $\dot u_\ii=0$, it implies $\dot{V}_\ii = 0, \ii=1,\dots,n_{\mathrm{gfm}}$ for all $\ii$. Therefore, the voltage equation~\eqref{eq:V_dot_expanded_final} achieves a steady-state.

\textbf{Step 7: Positivity.}
Since $V^\star\in \mathcal K_V$, we have
\[
V_\ii^\star\ge \Vnom-R_{V}>0,
\qquad \ii=1,\dots,n_{\mathrm{gfm}}.
\]
Thus, the steady state voltage is positive.

Hence, the constructed constant signals form a positive steady state of the voltage/filter/controller subsystem. 
\end{proof}

Having established the existence of a steady-state operating point, we now characterize its relation to the secondary-level control objectives. In particular, we show that the equilibrium induced by the proposed distributed secondary-level controller enforces the desired coordination among GFM-IIDGs, achieving voltage regulation and equal per-unitized reactive power sharing.

\begin{theorem}\label{thm:steady_state}
Voltage dynamics~\eqref{eq:V_dyan} under the secondary-level control~\eqref{eq:U_i_design_law} achieve equal per-unitized reactive power sharing and voltage regulation among GFM-IIDGs at steady state.
\end{theorem}
\begin{proof}
By Theorem~\ref{thm:steady_state_existence}, the closed-loop voltage dynamics admit a steady-state operating point. At this equilibrium, the measurements entering~\eqref{eq:gfm_optimization} are constant, so $\alpha_\ii(t_\s)$ and $D_\ii^{V_\Delta}(t_\s)$ are fixed at all steady-state sampling instants. Hence, problem~\eqref{eq:gfm_optimization} has the same optimizer, denoted $x^\sst$, at every such instant. By the sampled-data update law, after one update beyond steady state, $u_\ii(t_\s)=x_\ii^\sst$ for all sufficiently large $t_\s$ and all $\ii\in{1,\dots,n_{\mathrm{gfm}}}$. Consequently, $x_\ii^\sst-u_\ii(t_\s)=0$, and the interpolation law gives $\dot u_\ii(t)=0$ for all sufficiently large $t$. Since $U_\ii^\star$ differs from $u_\ii$ only by steady-state constant terms, $\dot U_\ii^\star=0$ for all $\ii\in{1,\dots,n_{\mathrm{gfm}}}$ at steady state. Setting $\dot{V}_\ii = 0$ in~\eqref{eq:V_dyan} and using  $\dot U_\ii^\star = 0, Q_\ii^{\mathrm{avg},\sst} = Q_\ii^{\sst}$ for all $\ii$, gives
\begin{align}\label{eq:Vi_steady_state}
\hspace{-0.08in} (V_\ii^{\sst} - \Vnom) =  u_\ii^{\sst} + \beta_{V_\ii} (\Vnom - V_\ii^{\sst}) - \tilde{\beta}_{Q_\ii} r_{V_\ii}Q_\ii^{\sst}.
\end{align}
Let $x^\sst$ be $\admm$ solution with accuracy $\varepsilon/4$. Then, 
\begin{align*}
&|V_\ii^{\sst} - \Vnom| = \big|u_\ii^{\sst} + \beta_{V_\ii} (\Vnom - V_\ii^{\sst})- \tilde{\beta}_{Q_\ii} r_{V_\ii}Q_\ii^{\sst}\big| \\
& \leq |x^\sst + \beta_{V_\ii} (\Vnom - V_\ii^{\sst}) - \tilde{\beta}_{Q_\ii} r_{V_\ii}Q_\ii^{\sst}| + |u_\ii^{\sst} - x^\sst| \\
& \leq V_\Delta + \textstyle \sqrt{\Upsilon (0.75)^{\theta_{\varepsilon}}} \leq V_\Delta + \varepsilon/2,
\end{align*}
where, we used~\eqref{eq:admm_iter_comp}. From~\eqref{eq:Vi_steady_state} for all $\ii$,
\begin{align}\label{eq:steady_state}
(1 + \beta_{V_\ii})(V_\ii^{\sst} - \Vnom) + (1 + \beta_{Q_\ii}) r_{V_\ii}Q_\ii^{\sst} = u_\ii^{\sst}. 
\end{align}
Since the steady-state inputs $u_\ii^\sst$ are obtained from the $\admm$ solution $x_\ii^\sst$~\cite{ADMM_tac}, the agreement error satisfies
$|u_\ii^{\sst} - u_\jj^{\sst}| \leq \textstyle \sqrt{\Upsilon (0.75)^{\theta_{\varepsilon}}} \leq \varepsilon$ for all $\ii,\jj$. Thus, with suitable design parameters, the secondary-level inputs~\eqref{eq:U_i_design_law} achieve the objectives of tasks $\mathcal{T}_1$ and $\mathcal{T}_2$ at steady state. In particular, 
\begin{itemize}
    \item [(i)] Let $\beta_{V_\ii} = \beta_{V} = - 1 + \mu$, $0< \mu \ll 0.01$, $\beta_{Q_\ii} = \beta_{Q} = 0$, for all GFM-IIDGs $\ii$. Then from~\eqref{eq:steady_state},  $|r_{V_\ii} Q_\ii^{\sst} - r_{V_\jj} Q_\jj^{\sst}| \approx |u_\ii^{\sst} - u_\jj^{\sst}| \leq \varepsilon \ \forall \ii,\jj$. Hence, equal per-unitized reactive power sharing is achieved among all GFM-IIDGs. 
    \item [(ii)] Let $\beta_{V_\ii} = \beta_{V} = 0, \beta_{Q_\ii} = \beta_{Q} = - 1 + \mu$, $0< \mu \ll 0.01$, for all GFM-IIDGs $\ii$. Then from~\eqref{eq:steady_state}, $|V_\ii^{\sst}  - V_\jj^{\sst}| = | (\Vnom + u_\ii^{\sst}) - (\Vnom + u_\jj^{\sst}) | \approx |u_\ii^{\sst} - u_\jj^{\sst}| \leq \varepsilon$ for all GFM-IIDGs $\ii,\jj$. Thus, the GFM-IIDGs buses have similar voltage magnitudes in steady state. Moreover, if $\beta_{V_\ii}=\beta_V=0$ and $\beta_{Q_\ii}=\beta_Q=-1+\mu$, with $0<\mu\ll0.01$, for all GFM-IIDGs $\ii$, then the solution $x_\ii^\sst$ of~\eqref{eq:gfm_optimization}, and hence $u_\ii^\sst$, satisfies $u_\ii^\sst\approx0$. Therefore, the steady-state voltage magnitude satisfies $V_\ii=V_{\mathrm{nom}}+u_\ii^\sst\approx\Vnom$.
\end{itemize}
Therefore, with appropriate design choices, the secondary-level controller achieves equal per-unitized reactive power sharing
and voltage regulation among GFM-IIDGs at steady state.
\end{proof}
\subsection{Analysis of GFM-IIDG Frequency Control Loop}
We next establish stability properties of the frequency dynamics~\eqref{eq:freq_dyan}. For the ease of analysis, here, we assume $n_\mathrm{gfm} = n_\mathrm{gfl}$. Let $\delta=[\delta_1,\dots,\delta_{n_\mathrm{gfm}}]^\top\in\mathbb{R}^{n_\mathrm{gfm}}$ and $\omega=[\omega_1,\dots,\omega_{n_\mathrm{gfm}}]^\top\in\mathbb{R}^{n_\mathrm{gfm}}$ denote the system states. With $P_\ii$ and $P_\ii^\mathrm{gfl}$ given by~\eqref{eq:activepower} and~\eqref{eq:gfl_power}, respectively, the dynamics~\eqref{eq:freq_dyan} can be written compactly as
\begin{align}\label{eq:phase_omega_closedloop_vector}
    \dot{\begin{bmatrix} \delta \\ \omega \end{bmatrix}} &= \mathbf{A} \begin{bmatrix} \delta \\ \omega \end{bmatrix} + \mathbf{B} \begin{bmatrix} 0 \\ \Omega \end{bmatrix}  + \mathbf{C} \begin{bmatrix} 0 \\ \boldsymbol{p} \end{bmatrix} + \mathbf{D} \begin{bmatrix} 0 \\ \overline{V} \end{bmatrix},
\end{align}
where, $\Omega = [\Omega_1, \dots, \Omega_{n_\mathrm{gfm}}]^\top \in \mathbb{R}^{n_\mathrm{gfm}}$ with $\Omega_\ii :=   (\omega_\mathrm{nom} + r_{\omega_\ii} P_\ii^{\min})/\tau_{P_\ii}$, $\boldsymbol{p} = [p_1, \dots, p_{n_\mathrm{gfm}}]^\top \in \mathbb{R}^{n_\mathrm{gfm}}$, $\overline{V} = [V^2_1, \dots, V^2_{n_\mathrm{gfm}}]^\top \in \mathbb{R}^{n_\mathrm{gfm}}$ and matrices $\mathbf{A},\mathbf{B}, \mathbf{C}, \mathbf{D}$ are
\begin{align}
    \mathbf{A} &:= \begin{bmatrix}
     0_{n_\mathrm{gfm}}  \ \mathbb{I}_{n_\mathrm{gfm}} \\
     \hspace{-0.12in} \mathbf{G}_{\omega} \ \hspace{0.12in} \boldsymbol{\tau}_P 
    \end{bmatrix}
    \mathbf{B}:= \begin{bmatrix}
     0_{n_\mathrm{gfm}} \ 0_{n_\mathrm{gfm}} \\
     0_{n_\mathrm{gfm}}  \ \mathbb{I}_{n_\mathrm{gfm}} 
    \end{bmatrix}
    \mathbf{C} := \begin{bmatrix}
     0_{n_\mathrm{gfm}} \ 0_{n_\mathrm{gfm}} \\
     \hspace{-0.18in} 0_{n_\mathrm{gfm}} \hspace{0.05in} \boldsymbol{P}
    \end{bmatrix} \nonumber \\
    \mathbf{D} &:= \begin{bmatrix}
     0_{n_\mathrm{gfm}} & 0_{n_\mathrm{gfm}} \\
     0_{n_\mathrm{gfm}} & \text{diag}((-G_{\ii \ii} r_{\omega_\ii})/\tau_{P_\ii})
    \end{bmatrix}, \label{eq:ABC_matrix_omega_delta}
\end{align}  
where, $\boldsymbol{\tau}_{P} := \text{diag}(- 1/\tau_{P_\ii}), [\mathbf{G}_\omega]_{\ii \jj} := -(r_{\omega_\ii}B_{\ii \jj} V_{\ii} V_{\jj})/\tau_{P_\ii}, $ $ [\mathbf{G}_\omega]_{\ii \ii} := -\sum_{\jj} [\mathbf{G}_\omega]_{\ii \jj}, \boldsymbol{P} := \text{diag}(r_{\omega_\ii}(P_\ii^{\max} - P_\ii^{\min})/\tau_{P_\ii})$. Since the phase angles are invariant under uniform shifts, the matrix $\mathbf{G}_\omega$ has a marginal mode associated with the vector $\mathbf{1}_{n_\mathrm{gfm}}$. Hence, stability of~\eqref{eq:phase_omega_closedloop_vector} is studied in relative coordinates. Let $\xi\in\mathbb{R}^{n_\mathrm{gfm}}_{>0}$ denote a normalized left null vector of $\mathbf{G}_\omega$, i.e., $\xi^\top \mathbf{G}_\omega=0$ and $\xi^\top\mathbf{1}_{n_\mathrm{gfm}}=1$, and define the projection matrix
\[
\Pi:=\mathbb{I}_{n_\mathrm{gfm}} -\mathbf{1}_{n_\mathrm{gfm}}\xi^\top .
\]
The projected phase and frequency variables are given by
\[
\tilde{\delta}:=\Pi\delta,\qquad
\tilde{\omega}:=\Pi\omega .
\]
Because $\mathbf{G}_\omega\mathbf{1}_{n_\mathrm{gfm}}=0$, the active-power coupling depends only on the relative phase angle, that is, $\mathbf{G}_\omega\delta=\mathbf{G}_\omega\tilde{\delta}$. Applying the projection to~\eqref{eq:phase_omega_closedloop_vector} yields
\[
\begin{aligned}
\dot{\tilde{\delta}} &= \tilde{\omega},\\
\dot{\tilde{\omega}}
&= \mathbf{G}_\omega\tilde{\delta}
+ \Pi\boldsymbol{\tau}_P \tilde{\omega}
+ \Pi\boldsymbol{\tau}_P\mathbf{1}_{n_\mathrm{gfm}}\omega_\xi
+ \Pi\Omega
+ \Pi\boldsymbol{P}\boldsymbol{p}\\
& \hspace{0.15in} + \Pi\operatorname{diag}\left(\frac{-G_{\ii\ii}r_{\omega_\ii}}{\tau_{P_\ii}}\right)\overline{V},
\end{aligned}
\]
where $\omega_\xi:=\xi^\top\omega$ denotes the weighted average frequency component. In particular, if the active-power filter time constants are identical, i.e., $\tau_{P_\ii}=\tau_P$ for all $\ii$, then $\boldsymbol{\tau}_P=-(1/\tau_P)\mathbb{I}_{n_\mathrm{gfm}}$ and the term $\Pi\boldsymbol{\tau}_P\mathbf{1}_{n_\mathrm{gfm}}\omega_\xi$ vanishes. In this case, the projected dynamics close in $(\tilde{\delta},\tilde{\omega})$ as
\begin{equation} \label{eq:centered_freq_delta}
\begin{aligned}
\hspace{-0.05in}\dot{\tilde{\delta}} &= \tilde{\omega}, \\
\hspace{-0.05in}\dot{\tilde{\omega}}
&= \textstyle \mathbf{G}_\omega\tilde{\delta}
-\frac{1}{\tau_P}\tilde{\omega}
+ \Pi\left[\Omega
+ \boldsymbol{P}\boldsymbol{p}+\operatorname{diag}\left(\frac{-G_{\ii\ii}r_{\omega_\ii}}{\tau_{P_\ii}}\right)\overline{V}\right].
\end{aligned}
\end{equation}

Thus, the marginal absolute-angle direction is removed, and the stability analysis can be carried out on the disagreement subspace. Let $\tilde{x}:=\begin{bmatrix} \tilde{\delta}\\ \tilde{\omega}\end{bmatrix}.$ The projected dynamics~\eqref{eq:centered_freq_delta} can be expressed in a compact form as
\begin{equation}\label{eq:projected_compact}
  \begin{aligned}
    \dot{\tilde{x}} &= \mathbf{A}_\Pi \tilde{x} + \mathbf{B}_\Pi d(t), \\
    \mathbf{A}_\Pi &:=
\begin{bmatrix}
0_{n_\mathrm{gfm}} & \mathbb{I}_{n_\mathrm{gfm}} \\
\mathbf{G}_\omega & -\frac{1}{\tau_P}\mathbb{I}_{n_\mathrm{gfm}}
\end{bmatrix},
\
\mathbf{B}_\Pi
:=
\begin{bmatrix}
0_{n_\mathrm{gfm}}\\
\Pi
\end{bmatrix}, \\
d(t)
&:=
\Omega
+
\boldsymbol{P}\boldsymbol{p}(t)
+
\operatorname{diag}\left(\frac{-G_{\ii\ii}r_{\omega_\ii}}{\tau_P}\right)\overline{V}(t).
\end{aligned}  
\end{equation}
We have the following result.

\begin{theorem}\label{thm:isps_projected_frequency}
Consider dynamics in~\eqref{eq:projected_compact}. Define the disagreement subspace as $\mathcal{S}
:=
\left\{
\tilde{x}\in\mathbb{R}^{2n_\mathrm{gfm}}
|
\xi^\top\tilde{\delta}=0,
\xi^\top\tilde{\omega}=0
\right\},
$
where $\xi$ is a normalized left null vector of $\mathbf{G}_\omega$. Suppose the conditions of Lemma~\ref{lem:A_pi_hurwitz} hold, so that $\mathbf{A}_\Pi$ restricted to $\mathcal{S}$ is Hurwitz. The projected frequency dynamics are input-to-state stable on $\mathcal{S}$ with respect to the effective input $\Pi d(t)$. In particular, there exist constants $\kappa \geq 1$, $\lambda>0$, and $\gamma>0$ such that
\[
\|\tilde{x}(t)\|
\leq
\kappa e^{-\lambda t}\|\tilde{x}(0)\|
+
\gamma
\sup_{0\leq s\leq t}\| \Pi d(s)\|,
\qquad t\geq 0.
\]
Equivalently, if the projected input satisfies $\sup_{t\geq 0}\|\Pi d(t)\|\leq \Delta_d$,
then
\[
\|\tilde{x}(t)\|
\leq
\kappa e^{-\lambda t}\|\tilde{x}(0)\|
+ \gamma \Delta_d,
\qquad t\geq 0.
\]
Therefore, the projected phase-frequency dynamics remain ultimately bounded, with the ultimate bound proportional to the size of the  projected control input.
\end{theorem}

\begin{proof}

First note that $\mathcal{S}$ is invariant under~\eqref{eq:projected_compact}. Indeed, if $\tilde{x}\in\mathcal{S}$, then $\xi^\top\dot{\tilde{\delta}}
=
\xi^\top\tilde{\omega}=0,
$
and using $\xi^\top\mathbf{G}_\omega=0$ and $\xi^\top\Pi=0$,
\[
\xi^\top\dot{\tilde{\omega}}
=
\xi^\top\mathbf{G}_\omega\tilde{\delta}
-
\frac{1}{\tau_P}\xi^\top\tilde{\omega}
+
\xi^\top\Pi d(t)
=0.
\]
Hence trajectories initialized in $\mathcal{S}$ remain in $\mathcal{S}$. Note that under the assumptions in Lemma~\ref{lem:A_pi_hurwitz},  $\mathbf{A}_\Pi$ is Hurwitz, and thus for any symmetric positive definite matrix $\mathbf{Q}$ on $\mathcal{S}$, there exists a symmetric positive definite matrix $\mathbf{M}$ on $\mathcal{S}$ such that $\mathbf{A}_\Pi^\top\mathbf{M}+\mathbf{M}\mathbf{A}_\Pi=-\mathbf{Q}$ on $\mathcal{S}$. Consider the Lyapunov function
\[
W(\tilde{x})=\tilde{x}^\top \mathbf{M}\tilde{x}.
\]
Along the trajectories of the projected system,
\begin{align*}
\dot{W}
&= \tilde{x}^\top
\left(
\mathbf{A}_\Pi^\top\mathbf{M}
+
\mathbf{M}\mathbf{A}_\Pi
\right)
\tilde{x}
+
2\tilde{x}^\top \mathbf{M}\mathbf{B}_\Pi d(t)\\
&=
-\tilde{x}^\top\mathbf{Q}\tilde{x}
+
2\tilde{x}^\top \mathbf{M}\mathbf{B}_\Pi d(t).
\end{align*}
Using Young's inequality, for any $\eta\in(0,1)$,
\[
2\tilde{x}^\top \mathbf{M}\mathbf{B}_\Pi d(t)
\leq
\eta\lambda_{\min}(\mathbf{Q})\|\tilde{x}\|^2
+
\frac{\|\mathbf{M}\|^2}{\eta\lambda_{\min}(\mathbf{Q})}
\|\Pi d(t)\|^2.
\]
Thus,
\[
\dot{W}
\leq
-(1-\eta)\lambda_{\min}(\mathbf{Q})\|\tilde{x}\|^2
+
\frac{\|\mathbf{M}\|^2}{\eta\lambda_{\min}(\mathbf{Q})}
\|\Pi d(t)\|^2.
\]
Since, for all $\tilde{x} \in \mathcal{S}$ we have, 
\[
\lambda_{\min}(\mathbf{M})\|\tilde{x}\|^2
\leq
W(\tilde{x})
\leq
\lambda_{\max}(\mathbf{M})\|\tilde{x}\|^2,
\]
there exist constants $c_1>0$ and $c_2>0$ such that
\[
\dot{W}
\leq
-c_1 W
+
c_2\|\Pi d(t)\|^2.
\]
Applying the comparison lemma gives
\[
W(t)
\leq
e^{-c_1 t}W(0)
+
\frac{c_2}{c_1}
\sup_{0\leq s\leq t}\|\Pi d(s)\|^2.
\]
Using the quadratic bounds on $W$ yields
\begin{align*}
\|\tilde{x}(t)\|
&\leq 
\sqrt{\frac{\lambda_{\max}(\mathbf{M})}{\lambda_{\min}(\mathbf{M})}}
e^{\frac{-c_1}{2}t}\|\tilde{x}(0)\|
\\
& \hspace{0.1in} +
\sqrt{\frac{c_2}{c_1\lambda_{\min}(\mathbf{M})}}
\sup_{0\leq s\leq t}\|\Pi d(s)\|.
\end{align*}
Thus, $\kappa := \sqrt{\frac{\lambda_{\max}(\mathbf{M})}{\lambda_{\min}(\mathbf{M})}}, \lambda = \frac{c_1}{c_2}, \gamma := \sqrt{\frac{c_2}{c_1\lambda_{\min}(\mathbf{M})}}$. Therefore, the projected dynamics are input-to-state stable with respect to the projected input $\Pi d(t)$.
\end{proof}

\begin{theorem}\label{thm:isps_full_frequency}
Consider the phase-frequency dynamics whose projected form is given by~\eqref{eq:projected_compact}. Suppose the conditions of Lemma~\ref{lem:A_pi_hurwitz} hold and $\mathcal{S}$ is the subspace as defined in Theorem~\ref{thm:isps_projected_frequency}. Define the weighted average-frequency error $e_\xi(t):=\xi^\top\omega(t)-\omega_{\rm nom}$. Then $e_\xi(t)$ satisfies $\dot{e}_\xi(t)
= -\frac{1}{\tau_P}e_\xi(t)
+
d_\xi(t),$
where $d_\xi(t) := \xi^\top d(t)-\frac{1}{\tau_P}\omega_{\rm nom}$.
Hence, $ |e_\xi(t)|
\leq
e^{-t/\tau_P}|e_\xi(0)|
+
\tau_P
\sup_{0\leq s\leq t}|d_\xi(s)|$.
Finally, the original frequency vector satisfies the decomposition
$\omega(t)-\omega_{\rm nom}\mathbf{1}_{n_\mathrm{gfm}} = \tilde{\omega}(t)+\mathbf{1}_{n_\mathrm{gfm}} e_\xi(t)$. Therefore,
\begin{align*}
\|\omega(t)-\omega_{\rm nom}\mathbf{1}_{n_\mathrm{gfm}}\|
& \textstyle \leq
\kappa e^{-\lambda t}\|\tilde{x}(0)\|
+ \gamma
\sup_{0\leq s\leq t}\|\Pi d(s)\|
\\ 
& \textstyle \hspace{-0.8in}+
\sqrt{n_\mathrm{gfm}} e^{-t/\tau_P}|e_\xi(0)| + \sqrt{n_\mathrm{gfm}} \tau_P
\sup_{0\leq s\leq t}|d_\xi(s)|.
\end{align*}
Equivalently, if $ \sup_{t\geq 0}\|\Pi d(t)\|\leq \Delta_d,
\
\sup_{t\geq 0}|d_\xi(t)| \leq \Delta_\xi$,
then
\begin{align*}
\|\omega(t)-\omega_{\rm nom}\mathbf{1}_{n_\mathrm{gfm}}\|
&\leq
\kappa e^{-\lambda t}\|\tilde{x}(0)\| 
+
\sqrt{n_\mathrm{gfm}} e^{-t/\tau_P}|e_\xi(0)| \\
& \hspace{0.2in} +
\gamma\Delta_d
+
\sqrt{n_\mathrm{gfm}}\tau_P\Delta_\xi.
\end{align*}
Thus, the frequency deviation from nominal synchronization is input-to-state practically stable~\cite{jiang1994small} with respect to the projected input $\Pi d(t)$ and the weighted average-frequency mismatch $d_\xi(t)$.
\end{theorem}

\begin{proof}
We prove the result in three steps. First, we show that the projected state remains in the disagreement subspace. Second, we apply the ISS estimate for the projected dynamics on this subspace from Theorem~\ref{thm:isps_projected_frequency}. Third, we combine this estimate with the scalar weighted-average frequency dynamics.

For any initial condition $(\delta(0),\omega(0))$, projected variables $(\tilde{\delta}(0), \tilde{\omega}(0))$ satisfy $\xi^\top\tilde{\delta}(0)
= \xi^\top\Pi\delta(0)
=0, \xi^\top\tilde{\omega}(0)=
\xi^\top\Pi\omega(0)
=0.$
Thus,
$\tilde{x}(0)=
\begin{bmatrix}
\tilde{\delta}(0) \ \
\tilde{\omega}(0)
\end{bmatrix}
\in\mathcal S.$

As shown in Theorem~\ref{thm:isps_projected_frequency}, $\mathcal{S}$ is invariant under the projected dynamics~\eqref{eq:projected_compact}.
Therefore, if $\tilde{x}(0)\in\mathcal S$, then $\tilde{x}(t)\in\mathcal S$ for all $t\geq 0$. Hence, the projected trajectory evolves entirely on the disagreement subspace, where the restriction of $\mathbf A_\Pi$ in~\eqref{eq:projected_compact} is Hurwitz by Lemma~\ref{lem:A_pi_hurwitz}.

Consequently, the ISS estimate on $\mathcal{S}$ applies to the projected system. Thus, from Theorem~\ref{thm:isps_projected_frequency}, there exist constants $\kappa\geq 1$, $\lambda>0$, and $\gamma>0$ such that
\[
\|\tilde{x}(t)\|
\leq
\kappa e^{-\lambda t}\|\tilde{x}(0)\|
+
\gamma\sup_{0\leq s\leq t}\|\Pi d(s)\|,
\quad t\geq 0,
\]
where, $\Pi d$ is the projected input. Since $\tilde{\omega}$ is a component of $\tilde{x}$, it follows that
\[
\|\tilde{\omega}(t)\|
\leq
\|\tilde{x}(t)\|
\leq
\kappa e^{-\lambda t}\|\tilde{x}(0)\|
+
\gamma\sup_{0\leq s\leq t}\|\Pi d(s)\|.
\]

It remains to analyze the weighted-average frequency component. Define $\omega_\xi(t):=\xi^\top\omega(t),
\
e_\xi(t):=\omega_\xi(t)-\omega_{\rm nom}$. Using the original frequency dynamics
$\dot{\omega} = \mathbf{G}_\omega\delta
-\frac{1}{\tau_P}\omega
+ d(t)$, where, $d(t)$ is as in \eqref{eq:projected_compact}. Multiplying by $\xi^\top$, we obtain
\begin{align*}
\dot{\omega}_\xi
& \textstyle = \xi^\top\dot{\omega} =
\xi^\top\mathbf{G}_\omega\delta
-\frac{1}{\tau_P}\xi^\top\omega
+
\xi^\top d(t).
\end{align*}
Since $\xi^\top\mathbf{G}_\omega=0$, this reduces to
$ \dot{\omega}_\xi = 
-\frac{1}{\tau_P}\omega_\xi
+
\xi^\top d(t).$ 
Subtracting $\omega_{\rm nom}$ from both sides gives
\[
\dot e_\xi = \textstyle 
-\frac{1}{\tau_P}e_\xi
+
\left(
\xi^\top d(t)-\frac{1}{\tau_P}\omega_{\rm nom}
\right).
\]
Define, $ d_\xi(t):=
\xi^\top d(t)-\frac{1}{\tau_P}\omega_{\rm nom}$. Then $\dot e_\xi = -\frac{1}{\tau_P}e_\xi+d_\xi(t)$.
By the variation-of-constants formula,
\[ \textstyle 
e_\xi(t) = 
e^{-t/\tau_P}e_\xi(0)
+
\int_0^t e^{-(t-s)/\tau_P}d_\xi(s)ds.
\]
Taking absolute values yields
\begin{align*}
|e_\xi(t)|
&\leq \textstyle 
e^{-t/\tau_P}|e_\xi(0)|
+
\int_0^t e^{-(t-s)/\tau_P}|d_\xi(s)|ds \\
& \textstyle \leq
e^{-t/\tau_P}|e_\xi(0)|
+
\sup_{0\leq s\leq t}|d_\xi(s)|
\int_0^t e^{-(t-s)/\tau_P}ds \\
& \textstyle \leq
e^{-t/\tau_P}|e_\xi(0)|
+
\tau_P\sup_{0\leq s\leq t}|d_\xi(s)|.
\end{align*}
Finally, decompose the original frequency vector as
$ \omega
= \Pi\omega+\mathbf{1}_{n_\mathrm{gfm}}\xi^\top\omega = 
\tilde{\omega}+\mathbf{1}_{n_\mathrm{gfm}}\omega_\xi.$
Therefore,
\[
\omega-\omega_{\rm nom}\mathbf{1}_{n_\mathrm{gfm}} = 
\tilde{\omega}
+
\mathbf{1}_{n_\mathrm{gfm}}(\omega_\xi-\omega_{\rm nom}) = \tilde{\omega}
+
\mathbf{1}_{n_\mathrm{gfm}} e_\xi.
\]
Using the triangle inequality,
\begin{align*}
\|\omega(t)-\omega_{\rm nom}\mathbf{1}_{n_\mathrm{gfm}}\|
&\leq
\|\tilde{\omega}(t)\|
+
\|\mathbf{1}_{n_\mathrm{gfm}} e_\xi(t)\| \\
& =
\|\tilde{\omega}(t)\|
+
\sqrt{n_\mathrm{gfm}}|e_\xi(t)|.
\end{align*}
Substituting the bounds on $\tilde{\omega}(t)$ and $e_\xi(t)$ gives
\begin{align*}
\|\omega(t)-\omega_{\rm nom}\mathbf{1}_{n_\mathrm{gfm}}\|
&\leq
\kappa e^{-\lambda t}\|\tilde{x}(0)\|
+
\gamma
\sup_{0\leq s\leq t}\|\Pi d(s)\|\\
&
\hspace{-0.6in} +
\sqrt{n_\mathrm{gfm}}e^{-t/\tau_P}|e_\xi(0)|
+
\sqrt{n_\mathrm{gfm}}\tau_P
\sup_{0\leq s\leq t}|d_\xi(s)|.
\end{align*}
This proves the stated bound on the frequency. If, in addition,
$ \sup_{t\geq 0}\|\Pi d(t)\|\leq \Delta_d, \sup_{t\geq 0}|d_\xi(t)|\leq \Delta_\xi$,
then the preceding estimate immediately implies
\begin{align*}
\|\omega(t)-\omega_{\rm nom}\mathbf{1}_{n_\mathrm{gfm}}\|
&\leq
\kappa e^{-\lambda t}\|\tilde{x}(0)\|
+
\sqrt{n_\mathrm{gfm}}e^{-t/\tau_P}|e_\xi(0)| \\
& \hspace{0.2in} +
\gamma\Delta_d
+
\sqrt{n_\mathrm{gfm}}\tau_P\Delta_\xi.
\end{align*}
Thus, the frequency deviation from nominal synchronization is input-to-state practically stable with respect to the disagreement input $\Pi d(t)$ and the weighted average-frequency mismatch $d_\xi(t)$.
\end{proof}

\subsection{Supporting Results}
\begin{lemma}\label{lem:psd_proof}
 Let $\mathcal{B} \in \mathbb{R}^{n_\mathrm{gfm}\times n_\mathrm{gfm}}$ be a square matrix with $[\mathcal{B}]_{\ii\ii} := -B_{\ii \ii} = -(B^{\mathrm{sh}}_{\ii} + \sum_{\kk \in N_\mathrm{i}} B_\mathrm{ik}), [\mathcal{B}]_{\ii\kk} := B_{\ii\kk}, \kk \neq \ii$, then the matrix $\Sigma \mathcal{B} + \mathcal{B}^\top \Sigma$ with $\Sigma := \text{diag}((\tilde{\beta}_\ii r_{V_\ii}V_\ii)/2\tau^2_{Q_\ii})$ is positive semi-definite.
\end{lemma}

\begin{proof} For any $x \in \mathbb{R}^{n_\mathrm{gfm}}$ we have
\[
x^\top (\Sigma \mathcal{B} + \mathcal{B}^\top \Sigma)x = 2x^\top \Sigma \mathcal{B} x,
\]
since $x^\top \Sigma \mathcal{B} x$ is a scalar and equal to its transpose. Since, $V_\ii > 0$, for all $\ii$, thus, $[\Sigma]_{\ii \ii} > 0$ for all $\ii$. Let $y = \Sigma^{1/2} x$ (well-defined since $\Sigma$ is a positive diagonal matrix). Then $x = \Sigma^{-1/2} y$ and
\[
2 x^\top \Sigma \mathcal{B} x 
= 2 y^\top ( \Sigma^{-1/2} \mathcal{B} \Sigma^{-1/2} ) y.
\]
Define $\nu := \Sigma^{-1/2} \mathcal{B} \Sigma^{-1/2}$. Note that $\nu = \sigma^\top \mathcal{B} \sigma$ with $\sigma = \Sigma^{-1/2}$, so $\nu$ is congruent to $\mathcal{B}$, and $\nu \succeq 0$ if and only if $\mathcal{B} \succeq 0$. Therefore, $\Sigma \mathcal{B} + \mathcal{B}^\top \Sigma \succeq 0$ if and only if $\mathcal{B} \succeq 0$, which is indeed the case due to the Gershgorin Circle Theorem.
\end{proof}

\begin{lemma}\label{lem:continuous_sol}
Solution $x_\s$ of~\eqref{eq:gfm_optimization} is continuous in $(V,Q(V))$.  
\end{lemma}
\begin{proof}
Due to the consensus constraints, every feasible point has the form
$x_s=c\mathbf 1_{n_{\mathrm{gfm}}}$. Hence, the optimization reduces to the
scalar problem
\[
\min_{c\in\mathcal I(V,Q)}
\|c\mathbf 1_{n_{\mathrm{gfm}}}-\alpha(V,Q)\|_2^2,
\]
where
\begin{align*}
\mathcal I(V,Q)
&:=
\bigcap_{\ii=1}^{n_{\mathrm{gfm}}}
[m_\ii(V,Q)-V_\Delta, m_\ii(V,Q)+V_\Delta] \\
m_\ii(V,Q)
&:=
(1+\beta_{Q_\ii})r_{V_\ii}Q_\ii^{\mathrm{avg}}(t_\s)
- \beta_{V_\ii}(\Vnom - V_\ii (t_\s))
\end{align*}
Writing $\bar\alpha(V,Q):=
\frac{1}{n_{\mathrm{gfm}}}\mathbf 1^\top\alpha(V,Q),$
we have
\[
\|c\mathbf 1-\alpha(V,Q)\|_2^2
=
n_{\mathrm{gfm}}(c-\bar\alpha(V,Q))^2
+
\|\alpha(V,Q)-\bar\alpha(V,Q)\mathbf 1\|_2^2.
\]
Therefore, the unique optimizer is
\[
c_\s(V,Q)=\Pi_{\mathcal I(V,Q)}(\bar\alpha(V,Q)).
\]
Let
\begin{align*}
\ell(V,Q)&:=\max_\ii\{m_\ii(V,Q)-V_\Delta\},
\\
r(V,Q)&:=\min_\ii\{m_\ii(V,Q)+V_\Delta\}.
\end{align*}
Then $\mathcal I(V,Q)=[\ell(V,Q),r(V,Q)]$. Since each $m_\ii(V,Q)$ is
continuous, $\ell$ and $r$ are continuous. Moreover, $\bar\alpha$ is
continuous because $\alpha$ is continuous. Thus
\[
c_\s(V,Q)
=
\min\{\max\{\bar\alpha(V,Q),\ell(V,Q)\},r(V,Q)\}
\]
is continuous. Hence, the solution
\[
x_\s(V,Q)=c_\s(V,Q)\mathbf 1_{n_{\mathrm{gfm}}}
\]
is continuous on the domain where $\mathcal I(V,Q)$ is nonempty.
\end{proof}

\begin{lemma}\label{lem:A_pi_hurwitz}
Suppose that the following conditions hold:
\begin{enumerate}
\item $\tau_{P_\ii}=\tau_P>0 , r_{\omega_\ii}>0, V_\ii>0$ for all $\ii\in \{1,\dots,n_{\mathrm{gfm}}\}$
\item the GFM-IIDG interaction graph is connected.
\end{enumerate}
Let
\[
\mathbf{A}_\Pi
:=
\begin{bmatrix}
0_{n_\mathrm{gfm}} & \mathbb{I}_{n_\mathrm{gfm}}\\
\mathbf{G}_\omega & -\frac{1}{\tau_P}\mathbb{I}_{n_\mathrm{gfm}}
\end{bmatrix}.
\]
Then the restriction of $\mathbf{A}_\Pi$ to the disagreement subspace
\[
\mathcal{S}
:=
\left\{
\tilde{x}=
\begin{bmatrix}
\tilde{\delta}\\
\tilde{\omega}
\end{bmatrix}
\in\mathbb{R}^{2n_\mathrm{gfm}}
|
\xi^\top\tilde{\delta}=0,
\xi^\top\tilde{\omega}=0
\right\},
\]
where $\xi$ is a normalized left null vector of $\mathbf{G}_\omega$, is Hurwitz.
\end{lemma}
\begin{proof}
Define the positive diagonal matrix
\[
R_\omega
:=
\operatorname{diag}(r_{\omega_1},\dots,r_{\omega_{n_\mathrm{gfm}}}).
\]
Next, define the symmetric weighted Laplacian $L_V$ associated with the GFM-IIDG interaction graph by
\[
[L_V]_{\ii\jj}
=
B_{\ii\jj}V_\ii V_\jj,
\quad \ii\neq \jj, \ \mbox{and} \
[L_V]_{\ii\ii}
=
\textstyle -\sum_{\jj\neq \ii}B_{\ii\jj}V_\ii V_\jj .
\]
Since $B_{\ii\jj}=B_{\jj\ii}<0$ on every edge and $V_\ii>0$, $L_V$ is a symmetric positive semidefinite weighted Laplacian. Since the GFM-IIDG interaction graph is connected, $
\operatorname{ker}(L_V)=\operatorname{span}\{\mathbf{1}_{n_\mathrm{gfm}}\}.$ Using the definition of $\mathbf{G}_\omega$, we can write
\[
\mathbf{G}_\omega = -\frac{1}{\tau_P}R_\omega L_V.
\]
Because $R_\omega\succ 0$, the matrix $R_\omega L_V$ is similar to a symmetric positive semidefinite matrix.

Indeed,
\[
R_\omega L_V
=
R_\omega^{1/2}
\left(
R_\omega^{1/2}L_VR_\omega^{1/2}
\right)
R_\omega^{-1/2}.
\]
Thus, $R_\omega L_V$ has real nonnegative eigenvalues. Consequently, $\mathbf{G}_\omega$ has real nonpositive eigenvalues. Since the graph is connected, $\mathbf{G}_\omega$ has one zero eigenvalue corresponding to the uniform-angle direction and all remaining eigenvalues are strictly negative. That is,
\[
\lambda_1(\mathbf{G}_\omega)=0,
\qquad
\lambda_k(\mathbf{G}_\omega)<0,
\quad k=2,\dots,n_\mathrm{gfm}.
\]
The zero eigenvalue corresponds to the direction $\mathbf{1}_{n_\mathrm{gfm}}$, which is removed by the projection. Hence, on the disagreement subspace $\mathcal{S}$, only the modes associated with $\lambda_k(\mathbf{G}_\omega)<0$ remain. 

To show that this implies the Hurwitz property of $\mathbf{A}_\Pi$ on the disagreement subspace, we now lift the modal properties of $\mathbf{G}_\omega$ to the second-order phase-frequency dynamics. Since $\mathbf{G}_\omega$ is similar to a symmetric matrix, it is diagonalizable and has real eigenvalues. Let $v_k$ be an eigenvector of $\mathbf{G}_\omega$ associated with a disagreement eigenvalue $\lambda_k<0$, so that
\[
    \mathbf{G}_\omega v_k=\lambda_k v_k .
\]
Consider an eigenpair 
$\left(s,
    \begin{bmatrix}
        \phi\\
        \psi
    \end{bmatrix}
    \right)$
of $\mathbf{A}_\Pi$ on the disagreement subspace. Then
\[
    \begin{bmatrix}
        0_{n_\mathrm{gfm}} & \mathbb{I}_{n_\mathrm{gfm}}\\
        \mathbf{G}_\omega & -\frac{1}{\tau_P}\mathbb{I}_{n_\mathrm{gfm}}
    \end{bmatrix}
    \begin{bmatrix}
        \phi\\
        \psi
    \end{bmatrix}
    =
    s
    \begin{bmatrix}
        \phi\\
        \psi
    \end{bmatrix}.
\]
The first block row gives $\psi=s\phi$. Substituting this relation into the second block row gives
\[
    \mathbf{G}_\omega\phi-\frac{1}{\tau_P}\psi=s\psi.
\]
Using $\psi=s\phi$, we obtain
\[
    \mathbf{G}_\omega\phi-\frac{s}{\tau_P}\phi=s^2\phi,
\]
or equivalently,
\[
    \mathbf{G}_\omega\phi
    =
    \left(s^2+\frac{s}{\tau_P}\right)\phi.
\]
Thus, for each eigenvalue $\lambda_k$ of $\mathbf{G}_\omega$, the corresponding eigenvalues $s$ of $\mathbf{A}_\Pi$ satisfy
\[
    s^2+\frac{1}{\tau_P}s-\lambda_k=0.
\]
On the disagreement subspace, $\lambda_k<0$. Hence,
\[
    \frac{1}{\tau_P}>0,
    \qquad
    -\lambda_k>0.
\]
Therefore, by the second-order Routh-Hurwitz criterion, the polynomial
\[
    s^2+\frac{1}{\tau_P}s-\lambda_k
\]
has both roots in the open left-half complex plane. Consequently, every eigenvalue of $\mathbf{A}_\Pi$ associated with a disagreement mode has a strictly negative real part.

The only eigenvalue of $\mathbf{G}_\omega$ that is not strictly negative is $\lambda_1=0$, which corresponds to the uniform-angle direction $\mathbf{1}_{n_\mathrm{gfm}}$. For this mode, the characteristic equation becomes
$s^2+\frac{1}{\tau_P}s=0,$
whose roots are $s=0, s=-\frac{1}{\tau_P}.$
The zero root corresponds to the absolute-angle mode, which is removed by projection onto the disagreement subspace. Therefore, no marginal mode remains on $\mathcal{S}$, and all eigenvalues of $\mathbf{A}_\Pi$ restricted to $\mathcal{S}$ have strictly negative real parts. Hence, $\mathbf{A}_\Pi$ restricted to $\mathcal{S}$ is Hurwitz.
\end{proof}

\bibliography{references}

\end{document}